\documentclass[10pt,journal]{IEEEtran}

\usepackage{amssymb}
\usepackage{amsmath}
\usepackage{cite}
\usepackage{url}
\usepackage{xcolor}
\usepackage{cite,graphicx,amsmath,amssymb}
\usepackage{fancyhdr}
\usepackage{mdwmath}
\usepackage{mdwtab}
\usepackage{caption}
\usepackage{amsthm}
\usepackage{setspace}
\usepackage{hyperref}
\usepackage{algorithm}
\usepackage{algorithmic}
\usepackage{pdfpages}
\usepackage{mathtools}
\usepackage{subcaption}
\hypersetup{colorlinks=false}

\graphicspath{{./src/img/}}

\newtheorem{remark}{Remark}
\newtheorem{theorem}{Theorem}

\newtheorem{lemma}{Lemma}

\newtheorem{corollary}{Corollary}

\newtheorem{proposition}{Proposition}

\title{STARS Enabled Integrated Sensing and Communications}
\author{

        Zhaolin~Wang,~\IEEEmembership{Graduate Student Member,~IEEE,}
        Xidong~Mu,~\IEEEmembership{Member,~IEEE,} \\
        and Yuanwei~Liu,~\IEEEmembership{Senior Member,~IEEE}\\
\thanks{Part of this paper has been presented at the 2022 IEEE 96th Vehicular Technology Conference: VTC2022-Fall, London and Beijing, Sep. 26-29, 2022 \cite{wang2022star}}
\thanks{The authors are with the School of Electronic Engineering
and Computer Science, Queen Mary University of London, London E1 4NS,
U.K. (e-mail: zhaolin.wang@qmul.ac.uk, xidong.mu@qmul.ac.uk, yuanwei.liu@qmul.ac.uk).}
\vspace*{-0.5cm}
}

\begin{document}

\maketitle
\vspace{-1.5cm}
\begin{abstract}
    A simultaneously transmitting and reflecting surface (STARS) enabled integrated sensing and communications (ISAC) framework is proposed, where the entire space is partitioned by STARS into a sensing space and a communication space. A novel sensing-at-STARS structure is proposed, where dedicated sensors are mounted at STARS to address the significant path loss and clutter interference of sensing. The Cram{\'e}r-Rao bound (CRB) of the two-dimensional (2D) direction-of-arrivals (DOAs) estimation of the sensing target is derived, which is then minimized subject to the minimum communication requirement. A novel approach is proposed to transform the complicated CRB minimization problem into a trackable modified Fisher information matrix (FIM) optimization problem. Both independent and coupled phase-shift models of STARS are investigated: 1) For the independent phase-shift model, to address the coupling problem of ISAC waveform and STARS coefficient, an efficient double-loop iterative algorithm based on the penalty dual decomposition (PDD) framework is conceived; 2) For the coupled phase-shift model, based on the PDD framework, a low complexity alternating optimization algorithm is proposed to tackle the coupled phase-shift constraint by alternately optimizing the amplitude and phase-shift coefficients of STARS with closed-form expressions. Finally, the numerical results demonstrate that: 1) STARS significantly outperforms conventional RIS in terms of CRB under the communication constraints; 2) The coupled phase-shift model achieves comparable performance to the independent one for low communication requirements or sufficient STARS elements; 3) It is more efficient to increase the number of passive elements of STARS than the active elements of the sensor; 4) Higher sensing accuracy can be achieved by STARS using the practical 2D maximum likelihood estimator compared with the conventional RIS.
\end{abstract}

\begin{IEEEkeywords}
    {C}ram{\'e}r-Rao bound, integrated sensing and communication (ISAC), simultaneously transmitting and reflecting intelligent surface (STARS).
\end{IEEEkeywords}

\section{Introduction}

Integrated sensing and communications (ISAC) technique has been recognized as an important enabler of next-generation wireless networks \cite{liu2022integrated}. The goal of ISAC is to integrate sensing function and communication function in a single platform to share the same resources, hardware facilities, and signal-processing modules. In addition to obvious advantages such as high spectrum, energy, and hardware efficiency, ISAC also enables wireless networks to be aware of the surrounding environment, thereby establishing ubiquitous intelligence in the future smart world \cite{zhang2021enabling}. This presents exciting opportunities for developing environment-aware technologies, including but not limited to augmented reality (AR), virtual reality (VR), and vehicle-to-everything (V2X). In addition, by replicating the physical world through sensing and exchanging information through communication, ISAC can connect the virtual world and the physical world, and build the foundation of the Metaverse \cite{xu2022full}.

Reconfigurable intelligent surface (RIS) is another promising technique to enable smart radio environment for next-generation wireless networks \cite{liu2021reconfigurable}. On the one hand, RIS is able to dynamically adjust its passive beamforming to enhance communication performance. On the other hand, RIS can also establish a line-of-sight (LoS) link with the target to facilitate sensing performance. However, the conventional transmitting/reflecting-only RIS requires sensing targets or communication users to be located on the same side of the RIS as the base station (BS), which can only achieve half-space coverage. To address this challenge, a promising simultaneous transmitting and reflecting surface (STARS) was recently proposed to realize $360^\circ$ full-space coverage \cite{liu2021star}. As a consequence, STARS is capable of providing new degrees of freedom (DoFs), i.e., both transmission and reflection beamforming \cite{mu2021simultaneously}, for enhancing both sensing and communication performance.

\subsection{Prior Works}

In recent years, there have been growing research interests in the ISAC from various perspectives. With the rapid development of multiple-input multiple-out (MIMO) techniques, multi-antenna arrays can provide higher DoFs to construct highly directional beams pointing to communication users and sensing targets, which motivates many works to investigate ISAC from the perspective of transmit beamforming \cite{liu2018mu, liu2020joint, pritzker2022transmit, zhang2022holographic}. Specifically, the authors of \cite{liu2018mu} jointly designed the transmit beamforming to approach the desired sensing beampattern while guaranteeing the minimum communication signal-to-interference-plus-noise-ratio (SINR) requirement. As a further advance, authors of \cite{liu2020joint} proposed to introduce an additional dedicated sensing signal to compensate for the DoF degradation caused by the limited number of communication users. To enhance the design flexibility between communication and sensing subsystems, several transmit beamforming designs were proposed in \cite{pritzker2022transmit} for guaranteeing the user-prescribed sensing and communication performance levels. Furthermore, in \cite{zhang2022holographic}, an ISAC system based on metamaterial antennas was conceived, which is capable of supporting ultra-dense radiation elements and realizing holographic beamforming.

Although the aforementioned transmit beamforming designs are capable of realizing a favorable tradeoff between sensing and communication, the sensing performance is not fully considered. It should be noted that the sensing function is carried out by transmitting a signal to the target and then analyzing the received echo signal reflected by the target, which indicates that the sensing performance is mainly characterized at the receiver, and it is not enough to only consider the beamforming of the transmitter. As a remedy, by considering the echo signals, some works adopted general sensing SINR \cite{chen2022generalized} or general sensing mutual information (MI) \cite{ouyang2022performance} as the performance metric, leading to joint transceiver designs. Meanwhile, from the perspective of estimation accuracy, the fundamental Cram{\'e}r-Rao bound (CRB), which characterizes the minimum achievable variance of the unbiased estimators at the receiver, has been considered in some recent works \cite{liu2021cramer, hua2022mimo}. In particular, the optimization framework for minimizing the CRB of the point target and the extended target was proposed in \cite{liu2021cramer}. Furthermore, the authors of \cite{hua2022mimo} investigated the optimal rate-CRB tradeoff region in ISAC systems.

Recently, motivated by the capability of RIS to adjust the signal propagation environment and significantly improve the performance of wireless communication \cite{huang2019reconfigurable, wu2019intelligent, yu2021irs}, RIS-assisted wireless sensing technique has received growing attention \cite{zhang2022metaradar, shao2022target, song2022intelligent, esmaeilbeig2022cramer}. In \cite{zhang2022metaradar}, the authors employed RIS to facilitate multi-target detection by jointly optimizing the sensing waveform and the RIS phase shifts. However, since the paths between BS, RIS, and targets have long round-trip distances and multiple hops, the performance of the RIS-assisted sensing system is limited by the significant path loss. Hence, the authors of \cite{shao2022target} proposed a RIS-self-sensing scheme, where the probing signal is sent by the co-located RIS controller and the echo signal is analyzed at the low-cost active sensors installed on the RIS. Furthermore, the authors of \cite{song2022intelligent} studied the RIS-assisted sensing from the CRB perspective, where the CRB for estimating the azimuth direction-of-arrival (DOA) of one target is derived and then minimized by jointly designing active and passive beamforming. As a further advance, a multi-RIS-aided sensing framework was proposed in \cite{esmaeilbeig2022cramer}, where the CRB of the estimated Doppler phase shift is minimized. There also have been extensive research contributions to RIS-assisted ISAC techniques. For example, a pair of joint active and passive beamforming designs were proposed in \cite{sankar2022beamforming} for ISAC systems assisted by single RIS and dual RISs, respectively. The authors of \cite{liu2022joint} jointly optimized the transmit waveform and passive beamforming to maximize the sensing SINR at the BS subject to different communication constraints. Moreover, RIS was exploited in \cite{xing2022passive} to enhance communication SINR and sensing detection resolution. To guarantee the estimation accuracy, the authors of \cite{wang2022joint} proposed to minimize the multi-user interference of communication under the minimum CRB requirement.

\subsection{Motivations and Challenges}

The motivations of this paper can be summarized in four folds.
\emph{Firstly}, in contrast to conventional RIS, STARS does not require the sensing target and the communicating user to be on the same side. Moreover, STARS can split a signal into two separate signals, which is a good match for the dual function in ISAC. Driven by the above observations, in this paper, we naturally propose adopting STARS to divide the whole space into two half-spaces, namely the \emph{sensing space} and the \emph{communication space}. Therefore, the signal from the BS is split at the STARS to carry out target sensing in the sensing space and serve communication users in the communication space, that is, ISAC is \emph{enabled} by the STARS. 

\emph{Secondly}, the sensing function is typically carried out at the BS by analyzing the echo signals from the target. However, the following challenges need to be addressed when utilizing STARS. On the one hand, the sensing signal has significant path loss over multiple hops, i.e., BS$\rightarrow$STARS$\rightarrow$target$\rightarrow$STARS$\rightarrow$BS, especially when the direct BS-target link is blocked and the STARS is deployed in the vicinity of the target. On the other hand, due to the transmission and reflection properties of STARS, the BS can also receive clutter signals from the communication space, which are difficult to distinguish from the desired echo signals from the sensing space. To address these challenges, we propose to install dedicated low-cost sensors on STARS following an idea similar to \cite{shao2022target}. Therefore, the sensing function can be carried out at these sensors, namely \emph{sensing-at-STARS}, rather than at the BS, thereby reducing the number of hops and path loss. Furthermore, to avoid clutter interference from the communication space, the side of the sensor facing the communication space can be physically blocked to achieve unidirectional reception of echo signals.

\emph{Thirdly}, most existing works assumed that the transmission and reflection phase shifts of STARS can be adjusted independently. However, this may be non-trivial to realize in practice since the active elements are required. Recently, a coupled phase-shift model is proposed for the STARS with purely passive lossless elements \cite{xu2022star}. Thus, to unveil the full potential of STARS in ISAC systems, we considered both independent and coupled phase-shift models of STARS.

\emph{Fourthly}, since the effectiveness of optimizing CRB has been widely demonstrated, we also considered it as the performance metric for sensing. However, in the existing works \cite{liu2021cramer, shao2022target, song2022intelligent, wang2022joint}, only the estimation of the azimuth DOA of the target is studied. Note that the RIS or STARS are usually equipped with a uniform planar array (UPA), which is capable of estimating both azimuth and elevation DOAs. This motivates us to investigate the CRB for the two-dimensional (2D) DOA estimation in this paper.

\subsection{Contributions}

The primary contributions of our paper are summarized as follows:

\begin{itemize}
    \item We propose a novel STARS-enabled ISAC system with a sensing-at-STARS structure, where the entire space is divided into a sensing space with a single target and a communication space with multiple users. Based on this setup, we derive the CRB for estimating the 2D DOAs of the target. Considering both independent and coupled phase-shift models of STARS, we formulate the CRB minimization problem under the communication SINR constraint.
    \item For the independent phase-shift model, instead of directly minimizing the complicated CRB, we propose an equivalent optimization for a more trackable modified Fisher information matrix (FIM). Furthermore, we propose an iterative algorithm based on penalty dual decomposition (PDD) to address the coupling of the optimization variables in the modified FIM.
    \item For the coupled phase-shift model, to tackle the additional complicated non-convex coupled phase-shift constraints, we conceive a low-complexity algorithm by iteratively optimizing the amplitude and phase-shift coefficients of STARS, where the optimal closed-form solutions are obtained in each iteration.
    \item Our numerical results verify the effectiveness of the proposed algorithms and reveal the superiority of the STARS over the conventional RIS in terms of ISAC performance both theoretically and practically. Some insights are also obtained from the numerical results. Firstly, the two phase-shift models of the STARS have similar performance in the cases of low communication requirements or sufficient STARS elements. Secondly, increasing the number of passive elements of STARS is more appealing than the number of sensor elements.
\end{itemize}

\subsection{Organization and Notations}

The rest of this paper is organized as follows. In Section \ref{sec:system_model}, the proposed STARS-enabled ISAC framework with the sensing-at-STARS structure is presented. Then, the CRB for the 2D DOA estimation is derived. In Section \ref{sec:indenependt}, a PDD-based algorithm is conceived to solve the CRB minimization problem with the independent phase-shift model, where a novel CRB simplification approach is proposed. In Section \ref{sec:coupled}, considering the coupled phase-shift model, another PDD-based algorithm is proposed to solve the corresponding CRB minimization problem, where the coupled phase-shift constraints are addressed in an alternating manner based on the optimal closed-form solutions. In Section \ref{sec:results}, the numerical results are provided to verify the effectiveness of the proposed framework and algorithms. Finally, this paper is concluded in Section \ref{sec:conclusion}

\emph{Notations:}
Scalars, vectors, and matrices are denoted by the lower-case, bold-face lower-case, and bold-face upper-case letters, respectively; 
$\mathbb{C}^{N \times M}$ and $\mathbb{R}^{N \times M}$ denotes the space of $N \times M$ complex and real matrices, respectively;
$a^*$ and $|a|$ denote the conjugate and magnitude of scalar $a$;  
$\mathbf{a}^H$ denotes the conjugate transpose of vector $\mathbf{a}$; 
$\mathrm{diag}(\mathbf{a})$ denotes a diagonal matrix with same value as the vector $\mathbf{a}$ on the diagonal;
$\mathbf{A} \succeq 0$ means that matrix $\mathbf{A}$ is positive semidefinite; 
$\mathrm{rank}(\mathbf{A})$ and $\mathrm{tr}(\mathbf{A})$ denote the rank and trace of matrix $\mathbf{A}$, respectively;
$\mathbb{E}[\cdot]$ denotes the statistical expectation; 
$\mathrm{Re}\{\cdot\}$ denotes the real component of a complex number;
$\mathcal{CN}(\mu, \sigma^2)$ denotes the distribution of a circularly symmetric complex Gaussian (CSCG) random variable with mean $\mu$ and variance $\sigma^2$.

\section{System Model} \label{sec:system_model}

As shown in Fig. \ref{fig:system_model}, we consider a narrowband STARS-enabled ISAC system, where the BS is equipped with a uniform linear array (ULA) consisting of $M$ antennas and the STARS is equipped with a uniform planar array (UPA) consisting of $N$ passive transmission-reflection (T-R) elements denoted by the set $\mathcal{N}$. The whole space is divided into two half-spaces by STARS, namely the \emph{sensing space} and the \emph{communication space}. Without loss of generality, we assume that the sensing space is on the reflection side and the communication space is on the transmission side. There is a single sensing target of interest in the sensing space and $K$ single-antenna communication users, whose indices are collected in $\mathcal{K}$, in the communication space. The direct links between BS and target/users are assumed to be blocked. To tackle the severe path-loss, we propose a sensing-at-STARS structure, where a dedicated low-cost sensor with a ULA consisting of $N_s$ elements is mounted on STARS. The side of the sensor facing the communication space is physically blocked to avoid the clutter signal from the communication space.  Furthermore, we consider a coherent time block of length $L$, during which the communication channels and sensing target parameters remain approximately constant.

\begin{figure}[t!]
    \centering
    \includegraphics[width=0.4\textwidth]{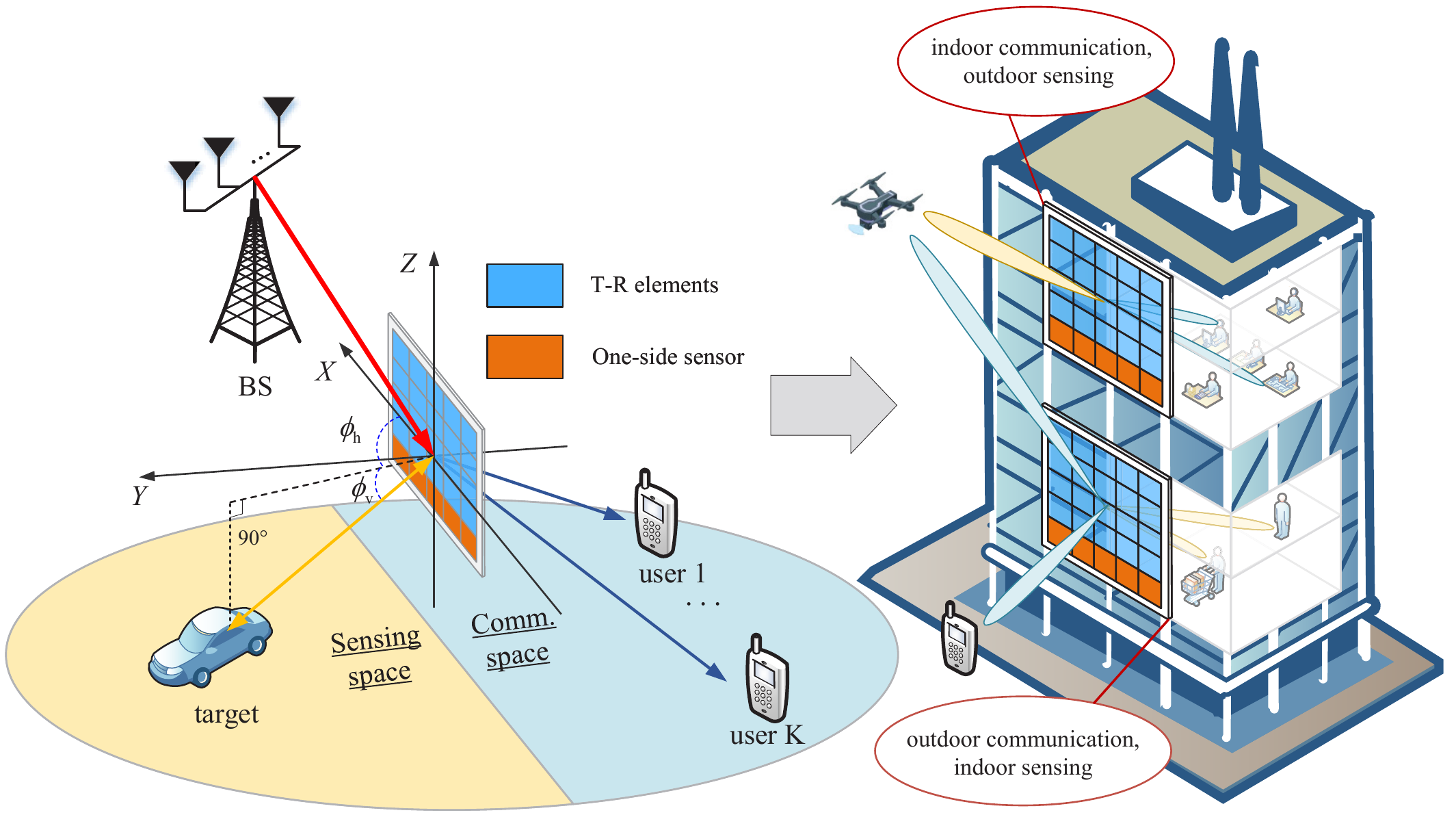}
    \caption{Illustration of the STARS-enabled ISAC system.}
    \label{fig:system_model}
\end{figure}

\subsection{STARS Model}
The energy splitting model is exploited to support simultaneous transmission and reflection of the STARS. In particular, the incident signal at the STARS from the BS is split into the sensing signal in the sensing space and the communication signal in the communication space. Denote $\mathbf{\Theta}_t \in \mathbb{C}^{N \times N}$ and $\mathbf{\Theta}_r \in \mathbb{C}^{N \times N}$ as the matrices of the transmission coefficients (TCs) and reflection coefficients (RCs), respectively, which can be modeled as 
\begin{align}
    \mathbf{\Theta}_i = \mathrm{diag} \left( \beta_{i,1} e^{j \varphi_{i,1}},\dots,\beta_{i,N} e^{j \varphi_{i,N}} \right), \forall i \in \{t, r\}.
\end{align}
In the above expression, $\beta_{i,n} \in [0,1]$ and $\varphi_{i,n} \in [0, 2\pi]$ denote the amplitude and phase-shift response of the $n$-th element. The exact value of $\beta_{i,n}$ and $\varphi_{i,n}$ is determined by the resistance and reactance of the STARS elements. In general, assuming that the phase shifts of TC and RC can be adjusted independently, the amplitudes need to satisfy the \emph{law of energy conservation} as follows:
\begin{equation}
    \beta_{t,n}^2 + \beta_{r,n}^2 = 1, \forall n \in \mathcal{N},
\end{equation}
which is termed as \emph{independent T\&R phase-shift} model. The independent T\&R phase-shift model requires the STARS elements to be active or lossy \cite{xu2022star}, which may result in higher manufacturing costs. As such, the STARS with low-cost passive and lossless elements has been studied in \cite{xu2022star}, where the electric and magnetic impedance of each element should be pure imaginary numbers. In this case, TCs and RCs also need to meet the following conditions \cite{xu2022star}:
\begin{equation}
    \cos(\varphi_{t,n} - \varphi_{r,n}) = 0, \forall n \in \mathcal{N}.
\end{equation}
Under the above constraints, the phase shifts of TCs and RCs are coupled, which is referred to as the \emph{coupled T\&R phase-shift} model. To fully investigate the role of STARS in the proposed system, both independent and coupled T\&R phase-shift models are considered in this paper.

\subsection{Signal Model}
To carry out sensing and communication simultaneously, the BS transmits the following joint signal at the time index $t$: 
\begin{equation}
    \mathbf{x}(t) = \mathbf{P}\mathbf{c}(t) + \mathbf{s}(t) = \sum_{k \in \mathcal{K}} \mathbf{p}_k c_k(t) + \mathbf{s}(t),
\end{equation}
where $\mathbf{P} = [\mathbf{p}_1,\dots,\mathbf{p}_K] \in \mathbb{C}^{M \times K}$ denotes the transmit beamforming matrix for delivering the information streams $\mathbf{c}(t) = [c_1(t),\dots,c_K(t)]^T \in \mathbb{C}^{K \times 1}$ to the $K$ communication users. The signal, $\mathbf{s}(t) \in \mathbb{C}^{M \times 1}$, is the dedicated sensing signal to achieve the full DoFs for target sensing \cite{liu2020joint}. The multiple beam transmission is exploited by the dedicated sensing signal. Thus, the covariance matrix of it, i.e., $\mathbf{R}_s = \mathbb{E}[\mathbf{s}(t) \mathbf{s}^H(t)]$, has a general rank. Furthermore, the communication signals are modeled as independent Gaussian random signals with zero mean and unit power, and the dedicated sensing signal is generated by pseudo-random coding, so that $\mathbb{E}[\mathbf{c}(t)\mathbf{c}(t)^H] = \mathbf{I}_{K}$ and $\mathbb{E}[\mathbf{c}(t)\mathbf{s}(t)^H] = \mathbf{0}_{K \times M}$. Thus, the covariance matrix of the transmit signal $\mathbf{x}(t)$ is given by 
\begin{equation}
    \mathbf{R}_x = \mathbb{E}[\mathbf{x}(t) \mathbf{x}(t)^H] = \mathbf{P} \mathbf{P}^H + \mathbf{R}_s.
\end{equation}
In practice, $\mathbf{R}_x$ can be calculated by averaging over $L$ time indexes as follows:
\begin{equation} \label{eqn:covariance_approx}
    \mathbf{R}_x \approx \frac{1}{L} \mathbf{X} \mathbf{X}^H,
\end{equation} 
where $\mathbf{X} = [\mathbf{x}(1),\dots,\mathbf{x}(L)]$. The above approximation is accurate when $L$ is large. Thus, we assume the accurate equality of \eqref{eqn:covariance_approx} is achieved throughout this paper. 

Given the transmit signal at the BS, the received signal at the $k$-th communication user in the communication space can be modeled as\footnote{Here, we omit the interference caused by the echo signal from the sensing space because its power at the communication user is negligible.}
\begin{align} \label{eqn:signal_communication}
    y_{c,k}(t) = &\underbrace{\mathbf{h}_k^H \mathbf{\Theta}_t \mathbf{G} \mathbf{p}_k c_k(t)}_{\text{desired signal}} + \underbrace{\sum_{i \in \mathcal{K} \backslash k} \mathbf{h}_k^H \mathbf{\Theta}_t \mathbf{G} \mathbf{p}_i c_i(t)}_{\text{inter-user interference}} \nonumber \\
    &+ \underbrace{\mathbf{h}_k^H \mathbf{\Theta}_t \mathbf{G} \mathbf{s}(t)}_{\text{sensing interference}} + n_k(t),
\end{align} 
where $\mathbf{G} \in \mathbb{C}^{N \times M}$ denotes the baseband channel matrix between the BS and the STARS, $\mathbf{h}_k \in \mathbb{C}^{N \times 1}$ denotes the baseband channel vector between the STARS and user $k$, and $n_{k} \sim \mathcal{CN}(0, \sigma_k^2)$ denotes the additive white Gaussian noise (AWGN) at user $k$ with the variance $\sigma_k^2$.

Then, in the sensing space, the received echo signal from the target over a coherent time block of length $L$ at the sensors can be modeled as 
\begin{equation} \label{eqn:echo_signal}
    \mathbf{Y}_s = \alpha \mathbf{b}(\phi_{h}, \phi_{v}) \mathbf{a}^T(\phi_{h}, \phi_{v}) \mathbf{\Theta}_r \mathbf{G} \mathbf{X} + \mathbf{N}_s,
\end{equation}
where $\alpha \in \mathbb{C}$ is the complex amplitude determined by the round-trip path-loss and the complex reflection factor of the target, $\phi_{h}$ and $\phi_{v}$ denotes azimuth and elevation DOAs of the target with respect to the STARS, respectively, $\mathbf{a}(\phi_{h}, \phi_{v}) \in \mathbb{C}^{N \times 1}$ denotes the steering vector of the STARS, $\mathbf{b}(\phi_{h}, \phi_{v}) \in \mathbb{C}^{N_s \times 1}$ denotes the steering vector of the sensors, and $\mathbf{N}_s$ denotes the AWGN noise with each entry obeying $\mathcal{CN}(0, \sigma_s^2)$. According to \cite{manikas2004differential}, the steering vectors under the planar wave assumption can be modeled as 
\begin{align}
    & \mathbf{a}(\phi_{h},\phi_{v}) = \mathrm{exp}(-j [\mathbf{r}_{X}, \mathbf{r}_{Y}, \mathbf{r}_{Z}] \mathbf{k}(\phi_{h},\phi_{v})), \\
    & \mathbf{b}(\phi_{h},\phi_{v}) = \mathrm{exp}(-j [ \bar{\mathbf{r}}_{X}, \bar{\mathbf{r}}_{Y}, \bar{\mathbf{r}}_{Z}] \mathbf{k}(\phi_{h},\phi_{v})),
\end{align}
where $[\mathbf{r}_{X}, \mathbf{r}_{Y}, \mathbf{r}_{Z}] \in \mathbb{R}^{N \times 3}$ and $[ \bar{\mathbf{r}}_{X}, \bar{\mathbf{r}}_{Y}, \bar{\mathbf{r}}_{Z}] \in \mathbb{R}^{N_s \times 3}$ have rows representing the Cartesian coordinates of the STARS elements and the sensors. The vector $\mathbf{k}(\phi_{h},\phi_{v}) \in \mathbb{R}^{3 \times 1}$ is the wavenumber vector defined as follows:
\begin{equation}
    \mathbf{k}(\phi_{h},\phi_{v}) = \frac{2 \pi}{\lambda_c}[\cos \phi_{h} \cos \phi_{v}, \sin \phi_{h} \cos \phi_{v}, \sin \phi_{v}]^T,
\end{equation}
where $\lambda_c$ denotes the wavelength of the carrier signal.  Without loss of generality, we assume that the ULA array of the sensors is deployed along the $X$-axis and the UPA array of the STARS is deployed within the $(X, Z)$ plane, i.e., $\bar{\mathbf{r}}_{Y} = \bar{\mathbf{r}}_{Z} = \mathbf{0}$ and $\mathbf{r}_{Y} = \mathbf{0}$. As a consequence, the steering vectors can be simplified as follows:
\begin{align}
    & \mathbf{a}(\phi_{h},\phi_{v}) = \mathrm{exp} \big( -j \frac{2 \pi}{\lambda_c} ( \mathbf{r}_{X} \cos \phi_{h} \cos \phi_{{v}} + \mathbf{r}_{Z} \sin \phi_{{v}} ) \big), \\
    & \mathbf{b}(\phi_{h},\phi_{v}) = \mathrm{exp} \big( -j \frac{2 \pi}{\lambda_c} \bar{\mathbf{r}}_{X} \cos \phi_{h} \cos \phi_{{v}} \big).
\end{align}
The echo signal model in \eqref{eqn:echo_signal} is assumed to be acquired in a specific range-Doppler bin. Thus, the range and Doppler parameters are omitted. Furthermore, in addition to the echo signal from the target, the sensors also receive the interference signal through the direct BS-sensor link. However, in practice, such interference can be effectively eliminated by the offline training \cite{shao2022target} and thereby is also omitted.

\subsection{Performance Metrics for Communication and 2D Sensing}
A common measure of multi-user communication is the SINR, which determines the achievable rate. Therefore, we focus on optimizing the SINR for each communication user. According to \eqref{eqn:signal_communication}, the SINR for decoding the desired signal 
at user $k$ is given by 
\begin{equation}  
    \gamma_k = \frac{|\mathbf{h}_k^H \mathbf{\Theta}_t \mathbf{G} \mathbf{p}_k|^2}{ \sum_{i \in \mathcal{K}\backslash k} |\mathbf{h}_k^H \mathbf{\Theta}_t \mathbf{G} \mathbf{p}_i|^2 \!+\! \mathbf{h}_k^H \mathbf{\Theta}_t \mathbf{G} \mathbf{R}_s \mathbf{G}^H \mathbf{\Theta}_t^H \mathbf{h}_k \!+\! \sigma_k^2}.
\end{equation}

For target sensing, we are interested in estimating the 2D DOAs $\phi_{h}$ and $\phi_{v}$ based on the observations $\mathbf{Y}_s$ over a coherent time block, which can be achieved by the classic maximum likelihood estimation (MLE) presented in Appendix A. After obtaining the estimated $\hat{\phi}_{h}$ and $\hat{\phi}_{v}$, the mean square error (MSE), i.e., $\epsilon_h^2 = \mathbb{E}[ |\phi_h - \hat{\phi}_h|^2 ]$ and $\epsilon_v^2 = \mathbb{E}[ |\phi_v - \hat{\phi}_v|^2 ]$, is commonly used to evaluate the estimation performance. However, the optimization of MSE is intractable since it is difficult to obtain closed-form expressions. Therefore, we are interested in optimizing the CRB for estimating the 2D DOAs $\phi_{h}$ and $\phi_{v}$, which provides a lower bound for the MSE and has the closed-form expression. To derive the CRB, we first vectorize the matrix $\mathbf{Y}_s$ as follows:
\begin{equation} \label{eqn:signal_vector}
    \mathbf{y}_s = \mathrm{vec}(\mathbf{Y}_s) = \mathbf{u} + \mathbf{n}_s,
\end{equation}
where $\mathbf{u} = \mathrm{vec}(\alpha  \mathbf{b}(\phi_{h},\phi_{v}) \mathbf{a}^T(\phi_{h},\phi_{v}) \mathbf{\Theta}_r \mathbf{G} \mathbf{X})$ and $\mathbf{n}_s = \mathrm{vec}(\mathbf{N}_s)$. Denote $\boldsymbol{\xi} = [\boldsymbol{\phi}^T, \tilde{\boldsymbol{\alpha}}^T]^T$, where $\boldsymbol{\phi} = [\phi_{h}, \phi_{v}]^T$ and $\tilde{\boldsymbol{\alpha}} = [\mathrm{Re}(\alpha), \mathrm{Im}(\alpha)]^T$, as the unknown parameters to estimate. Then, the FIM for estimating the vector $\boldsymbol{\xi}$ can be partitioned as  
\begin{equation}
    \mathbf{J}_{\boldsymbol{\xi}} = \begin{bmatrix}
        \mathbf{J}_{\boldsymbol{\phi} \boldsymbol{\phi}} &\mathbf{J}_{\boldsymbol{\phi} \tilde{\boldsymbol{\alpha}}} \\
        \mathbf{J}_{\boldsymbol{\phi} \tilde{\boldsymbol{\alpha}}}^T & \mathbf{J}_{\tilde{\boldsymbol{\alpha}} \tilde{\boldsymbol{\alpha}}}
    \end{bmatrix}.
\end{equation}
Then, the CRB matrix for estimating $\boldsymbol{\phi}$ is given by \cite{bekkerman2006target}
\begin{equation} \label{eqn:CRB}
    \mathrm{CRB}(\boldsymbol{\phi}) = \left [ \mathbf{J}_{\boldsymbol{\phi} \boldsymbol{\phi}} - \mathbf{J}_{\boldsymbol{\phi} \tilde{\boldsymbol{\alpha}}} \mathbf{J}_{\tilde{\boldsymbol{\alpha}} \tilde{\boldsymbol{\alpha}}}^{-1} \mathbf{J}_{\boldsymbol{\phi} \tilde{\boldsymbol{\alpha}}}^T \right ]^{-1},
\end{equation} 
where the expressions of the matrices $\mathbf{J}_{\boldsymbol{\phi} \boldsymbol{\phi}}$, $\mathbf{J}_{\boldsymbol{\phi} \tilde{\boldsymbol{\alpha}}}$, and $\mathbf{J}_{\tilde{\boldsymbol{\alpha}} \tilde{\boldsymbol{\alpha}}}$ are derived in Appendix~B. It can be observed that $\mathrm{CRB}(\boldsymbol{\phi})$ depends on the target DOAs $\boldsymbol{\phi}$. The knowledge of $\boldsymbol{\phi}$ is thus required to design the joint waveform at the BS and the TCs/RCs at the STARS. Note that in practice, the target DOAs typically do not change significantly between two adjacent coherent time blocks. Therefore, the estimated or predicted DOAs from estimation results in previous coherent time blocks are sufficient for system design \cite{liu2021cramer}. Based on the above analysis, we assume that the DOAs $\boldsymbol{\phi}$ are fixed for the optimization problem. Furthermore, it is assumed that the communication channels of both independent and coupled T\&R phase-shifted STARS have been obtained by the advanced channel estimation method \cite{wu2021channel}.

\section{CRB Optimization Design with Independent T\&R Phase-Shift} \label{sec:indenependt}

In this section, based on the proposed framework, we focus on the minimization of CRB in the case of STARS with independent T\&R phase shifts, which can be simplified to the optimization of a modified FIM. We then propose to solve the resulting optimization problem by invoking the PDD framework.

\subsection{Problem Formulation}
The diagonal elements of the CRB matrix represent the minimum variance of each target parameter estimated by the unbiased estimator \cite{kay1993fundamentals}. Therefore, we aim to minimize the trace of the CRB matrix, while guaranteeing the minimum SINR level at each communication user. To optimize the complicated CRB matrix $\mathrm{CRB}(\boldsymbol{\phi})$, we first transform it into a more trackable form according to the following proposition.
\begin{proposition} \label{proposition:CRB_minimize}
    \emph{
        Minimizing the trace of CRB matrix $\mathrm{CRB}(\boldsymbol{\phi})$ is equivalent to solving the following optimization problem
        \begin{subequations} \label{problem:proposition_1}
            \begin{align}
                \min_{\mathbf{U}, \mathbf{\Omega}} \quad & 
                \mathrm{tr}( \mathbf{U}^{-1} ) \\ 
                \label{constraint:Schur}
                \mathrm{s.t.} \quad & 
                \begin{bmatrix}
                    \mathbf{J}_{\boldsymbol{\phi} \boldsymbol{\phi}} - \mathbf{U} &\mathbf{J}_{\boldsymbol{\phi} \tilde{\boldsymbol{\alpha}}} \\
                    \mathbf{J}_{\boldsymbol{\phi} \tilde{\boldsymbol{\alpha}}}^T & \mathbf{J}_{\tilde{\boldsymbol{\alpha}} \tilde{\boldsymbol{\alpha}}}
                \end{bmatrix} \succeq 0, \\
                & \mathbf{U} \succeq 0,
            \end{align}
        \end{subequations}
        where $\mathbf{U} \in \mathbb{C}^{2 \times 2}$ is an auxiliary matrix and $\mathbf{\Omega}$ denotes the original optimization variables.
    }
\end{proposition}

\begin{proof}
    Firstly, we have that the FIM $\mathbf{J}_{\boldsymbol{\phi} \boldsymbol{\phi}} - \mathbf{J}_{\boldsymbol{\phi} \tilde{\boldsymbol{\alpha}}} \mathbf{J}_{\tilde{\boldsymbol{\alpha}} \tilde{\boldsymbol{\alpha}}}^{-1} \mathbf{J}_{\boldsymbol{\phi} \tilde{\boldsymbol{\alpha}}}^T$
    is a positive semidefinite matrix. According to \cite[Example 3.46]{boyd2004convex}, the function $\mathrm{tr}(\mathbf{A}^{-1})$ is matrix decreasing on the 
    positive semidefinite matrix space. Therefore, minimizing $\mathrm{tr} \left( [\mathbf{J}_{\boldsymbol{\phi} \boldsymbol{\phi}} - \mathbf{J}_{\boldsymbol{\phi} \tilde{\boldsymbol{\alpha}}} \mathbf{J}_{\tilde{\boldsymbol{\alpha}} \tilde{\boldsymbol{\alpha}}}^{-1} \mathbf{J}_{\boldsymbol{\phi} \tilde{\boldsymbol{\alpha}}}^T]^{-1} \right)$
    is equivalent to minimizing $\mathrm{tr}(\mathbf{U}^{-1})$, where $\mathbf{U} \succeq 0$, subject to the constraint
    \begin{equation} \label{eqn:FIM_constraint}
        \mathbf{J}_{\boldsymbol{\phi} \boldsymbol{\phi}} - \mathbf{J}_{\boldsymbol{\phi} \tilde{\boldsymbol{\alpha}}} \mathbf{J}_{\tilde{\boldsymbol{\alpha}} \tilde{\boldsymbol{\alpha}}}^{-1} \mathbf{J}_{\boldsymbol{\phi} \tilde{\boldsymbol{\alpha}}}^T \succeq \mathbf{U}.
    \end{equation}
    Then, based on the Schur complement condition \cite{zhang2005schur}, the constraint \eqref{constraint:Schur} with a modified FIM can be obtained.    
\end{proof}

\begin{remark}
    \emph{\textbf{Proposition \ref{proposition:CRB_minimize}} can be easily extended to the cases of multiple targets. For example, consider a system with $Q$ targets. Denote $\alpha_q$, $\phi_{q, \mathrm{h}}$, and $\phi_{q, \mathrm{v}}$ as the complex amplitude, the azimuth DOA, and the elevation DOA of the $q$-th target, respectively. The unknown parameters become $\bar{\boldsymbol{\xi}} = [ \bar{\boldsymbol{\phi}}^T, \bar{\tilde{\boldsymbol{\alpha}}}^T ]$, where $\bar{\boldsymbol{\phi}} = [\phi_{1,h}, \phi_{1,v},\dots,\phi_{Q,h}, \phi_{Q,v}]^T \in \mathbb{R}^{2Q \times 1}$ and $\bar{\tilde{\boldsymbol{\alpha}}} = [\mathrm{Re}(\alpha_1), \mathrm{Im}(\alpha_1),\dots,\mathrm{Re}(\alpha_Q), \mathrm{Im}(\alpha_Q)] \in \mathbb{R}^{2Q \times 1}$. Following the same path in Appendix B, it can be shown that the corresponding CRB matrix of DOA estimation is in the same form as \eqref{eqn:CRB} but with a different dimension of $2Q \times 2Q$. Therefore, by defining a new auxiliary semidefinite matrix $\bar{\mathbf{U}} \in \mathbb{C}^{2Q \times 2Q}$, the corresponding CRB minimization problem for multiple targets can be converted into the same form as problem \eqref{problem:proposition_1}.}
\end{remark}

Based on \textbf{Proposition \ref{proposition:CRB_minimize}}, the optimization problem for minimizing $\mathrm{CRB}(\boldsymbol{\phi})$ can be formulated as follows:
\begin{subequations}
    \begin{align}
        \mathcal{P}: \quad \min_{ \scriptstyle \mathbf{U}, \mathbf{P}, \mathbf{R}_s, \atop \scriptstyle \boldsymbol{\theta}_t, \boldsymbol{\theta}_r} \quad & 
        \mathrm{tr} \left( \mathbf{U}^{-1} \right), \\
        \label{constraint:P1_1}
        \mathrm{s.t.} \quad &
        \begin{bmatrix}
            \mathbf{J}_{\boldsymbol{\phi} \boldsymbol{\phi}} - \mathbf{U} &\mathbf{J}_{\boldsymbol{\phi} \tilde{\boldsymbol{\alpha}}} \\
            \mathbf{J}_{\boldsymbol{\phi} \tilde{\boldsymbol{\alpha}}}^T & \mathbf{J}_{\tilde{\boldsymbol{\alpha}} \tilde{\boldsymbol{\alpha}}}
        \end{bmatrix} \succeq 0, \\
        \label{constraint:P1_2}
        & \gamma_k \ge \overline{\gamma}_k, \forall k\\
        \label{constraint:P1_3}
        & \mathrm{tr}(\mathbf{P}\mathbf{P}^H + \mathbf{R}_s) \le P, \\
        \label{constraint:P1_4}
        & \beta_{t,n}^2 + \beta_{r,n}^2 = 1, 0 \le \beta_{t,n}, \beta_{r,n} \le 1, \forall n, \\
        \label{constraint:P1_5}
        & \mathbf{R}_s \succeq 0, \mathbf{U} \succeq 0,
    \end{align}
\end{subequations}
where $\boldsymbol{\theta}_i = [\beta_{i,1} e^{j \varphi_{i,1}},\dots,\beta_{i,N} e^{j \varphi_{i,N}}]^T, \forall i \in \{t,r\}, $ is the vector comprised by the entries on the diagonal of the matrix $\mathbf{\Theta}_i, \forall i \in \{t,r\}$, $\overline{\gamma}_k \ge 0$ in constraint \eqref{constraint:P1_2} denotes the minimum SINR threshold of user $k$, and $P \ge 0$ in constraint \eqref{constraint:P1_3} denotes the total power budget at the BS. Constraint \eqref{constraint:P1_4} characterizes the amplitude relationships between TCs and RCs of the STARS. Finally, constraint \eqref{constraint:P1_5} ensures the matrices $\mathbf{R}_s$ and $\mathbf{U}$ to be positive semidefinite. The main challenges for solving problem $\mathcal{P}$ is the highly coupled transmit waveform matrices, i.e., $\mathbf{P}$ and $\mathbf{R}_s$, and the STARS coefficient vectors, i.e., $\boldsymbol{\theta}_t$ and $\boldsymbol{\theta}_r$, in the non-convex constraints \eqref{constraint:P1_1} and \eqref{constraint:P1_2}. 
As a consequence, it is challenging to find the global optimum of problem $\mathcal{P}$. To the best of the authors' knowledge, there is no existing algorithm for solving the unique coupling problem in CRB minimization. Therefore, in the following subsections, we develop an efficient PDD-based algorithm to obtain a high-quality solution to problem $\mathcal{P}$.

\subsection{PDD Framework for Solving Problem $\mathcal{P}$}
As shown in \cite{shi2020penalty}, the standard PDD optimization framework is developed in a double-loop structure, where the augmented Lagrangian (AL) problem of the original problem is optimized in a block coordinate descent (BCD) manner in the inner loop, while the Lagrangian dual variables and penalty factors are updated in the outer loop. Therefore, the key idea of applying the PDD method for solving $\mathcal{P}$ is to construct an AL problem of it that has a simple or even closed-form solution at each step of the BCD. To this end, we first define auxiliary variables as follows:
\begin{align}
    \mathbf{F} = \mathbf{\Theta}_r \mathbf{G} \mathbf{R}_x \mathbf{G}^H \mathbf{\Theta}_r^H,
\end{align}
Then, problem $\mathcal{P}$ can be reformulated as
\begin{subequations} \label{problem:P_transform}
    \begin{align}
        \min_{\boldsymbol{\chi}} \quad & \mathrm{tr}(\mathbf{U}^{-1})    \\
        \label{constraint:P1.1_1}
        \mathrm{s.t.} \quad & \mathbf{F} = \mathbf{\Theta}_r \mathbf{G} \mathbf{R}_x \mathbf{G}^H \mathbf{\Theta}_r^H, \\
        \label{constraint:P1.1_2}
        &\begin{bmatrix} 
            \mathbf{J}_{\boldsymbol{\phi} \boldsymbol{\phi}}(\mathbf{F}) - \mathbf{U} &\mathbf{J}_{\boldsymbol{\phi} \tilde{\boldsymbol{\alpha}}}(\mathbf{F}) \\
            \mathbf{J}_{\boldsymbol{\phi} \tilde{\boldsymbol{\alpha}}}^T(\mathbf{F}) & \mathbf{J}_{\tilde{\boldsymbol{\alpha}} \tilde{\boldsymbol{\alpha}}}(\mathbf{F})
        \end{bmatrix} \succeq 0, \\
        & \eqref{constraint:P1_2} - \eqref{constraint:P1_5},
    \end{align}
\end{subequations}
where $\boldsymbol{\chi} \triangleq  \{\mathbf{U}, \mathbf{P}, \mathbf{R}_s, \mathbf{F}, \boldsymbol{\theta}_t, \boldsymbol{\theta}_r\}$ represents all optimization variables. The entries of the matrix in the left-hand side of constant \eqref{constraint:P1.1_2} are given by
\begin{align}
    &\mathbf{J}_{\boldsymbol{\phi} \boldsymbol{\phi}}(\mathbf{F}) = \frac{2|\alpha|^2 L}{\sigma_s^2} \mathrm{Re} \left( \begin{bmatrix}
        \mathrm{tr} (\dot{\mathbf{B}}_{\phi_{\mathrm{h}}} \mathbf{F} \dot{\mathbf{B}}_{\phi_{\mathrm{h}}}^H)  & \mathrm{tr}(\dot{\mathbf{B}}_{\phi_{\mathrm{h}}} \mathbf{F} \dot{\mathbf{B}}_{\phi_{\mathrm{v}}}^H)  \\
        \mathrm{tr}(\dot{\mathbf{B}}_{\phi_{\mathrm{h}}} \mathbf{F} \dot{\mathbf{B}}_{\phi_{\mathrm{v}}}^H) & \mathrm{tr}(\dot{\mathbf{B}}_{\phi_{\mathrm{v}}} \mathbf{F} \dot{\mathbf{B}}_{\phi_{\mathrm{v}}}^H)
    \end{bmatrix}\right), \\
    &\mathbf{J}_{\boldsymbol{\phi} \tilde{\boldsymbol{\alpha}}}(\mathbf{F}) = \frac{2 L}{\sigma_s^2} \mathrm{Re} \left( \begin{bmatrix}
        \alpha^* \mathrm{tr}\left( \mathbf{B} \mathbf{F} \dot{\mathbf{B}}_{\phi_{\mathrm{h}}}^H \right) \\
        \alpha^* \mathrm{tr}\left( \mathbf{B} \mathbf{F} \dot{\mathbf{B}}_{\phi_{\mathrm{v}}}^H \right)
    \end{bmatrix} [1, j] \right), \\
    &\mathbf{J}_{\tilde{\boldsymbol{\alpha}} \tilde{\boldsymbol{\alpha}}}(\mathbf{F}) = \frac{2 L}{\sigma_s^2} \mathbf{I}_2 \mathrm{tr} \left( \mathbf{B} \mathbf{F} \mathbf{B}^H \right).
\end{align}

\begin{algorithm}[tb]
    \caption{PDD-based algorithm for solving problem $\mathcal{P}$.}
    \label{alg:PDD}
    \begin{algorithmic}[1]
        \STATE{Initialize feasible $\boldsymbol{\chi}^{[0]}$, $\mathbf{\Upsilon}^{[0]}$, $\rho^{[0]} \ge 0$, and set $0 < c < 1$,  $n=1$.}
        \REPEAT
        \STATE{ $\boldsymbol{\chi}^{[n+1]} = \mathrm{optimize} \left(  \mathcal{P}_{\mathrm{AL}}(\rho^{[n]}, \mathbf{\Upsilon}^{[n]}) \right) $ }
        \IF{$h(\boldsymbol{\chi}^{[n+1]}) \le \eta^{[n]}$}
            \STATE{ $\mathbf{\Upsilon}^{[n+1]} = \mathbf{\Upsilon}^{[n]} + \frac{1}{\rho} ( \mathbf{F}^{[n+1]} - \mathbf{\Theta}_r^{[n+1]} \mathbf{G} \mathbf{R}_x^{[n+1]} $ $\times \mathbf{G}^H (\mathbf{\Theta}_r^{[n+1]})^H )$. }
            \STATE{ $\rho^{[n+1]} = \rho^{[n]}$.}
        \ELSE
            \STATE{ $\mathbf{\Upsilon}^{[n+1]} = \mathbf{\Upsilon}^{[n]}$, $\rho^{[n+1]} = c \rho^{[n]}$. }
        \ENDIF
        \STATE{$n=n+1$.}
        \UNTIL{the constraint violation $h(\boldsymbol{\chi})$ falls below a predefined threshold.}
    \end{algorithmic}
\end{algorithm}

\noindent By introducing the Lagrangian dual variable $\mathbf{\Upsilon} \in \mathbb{C}^{N \times N}$ and the penalty factor $\rho > 0$ for the equality constraints \eqref{constraint:P1.1_1}, the following AL problem of \eqref{problem:P_transform} can be obtained:
\begin{subequations} \label{problem:AL_problem}
    \begin{align}
        \mathcal{P}_{\mathrm{AL}}(\rho, \mathbf{\Upsilon}): \quad & \min_{\boldsymbol{\chi}} \quad 
        \mathrm{tr}(\mathbf{U}^{-1}) + P_\rho(\boldsymbol{\chi}, \mathbf{\Upsilon}) \\
        \mathrm{s.t.} \quad & \eqref{constraint:P1_2} - \eqref{constraint:P1_5}, \eqref{constraint:P1.1_2},
    \end{align}
\end{subequations}
where 
\begin{equation}
    P_\rho(\boldsymbol{\chi}, \mathbf{\Upsilon}) = \frac{1}{2\rho} \|\mathbf{F} - \mathbf{\Theta}_r \mathbf{G} \mathbf{R}_x \mathbf{G}^H \mathbf{\Theta}_r^H + \rho \mathbf{\Upsilon} \|^2.
\end{equation}
According to \cite{shi2020penalty}, the PDD-based algorithm for solving problem $\mathcal{P}$ is summarized in \textbf{Algorithm \ref{alg:PDD}}, where $h(\boldsymbol{\chi})$ is the constraint violation function defined as follows:
\begin{equation}
    h(\boldsymbol{\chi}) = \|\mathbf{F} - \mathbf{\Theta}_r \mathbf{G} \mathbf{R}_x \mathbf{G}^H \mathbf{\Theta}_r^H \|_{\infty}
\end{equation} 
Moreover, $\{\eta^{[n]}\}_{n=1}^\infty$ is a sequence that converges to zero, which can be set empirically. In this paper, we set $\eta^{[n]} = 0.99h(\boldsymbol{\chi}^{[n-1]})$.  It can be observed that when the penalty factor $\rho$ is sufficiently small, the penalty term $P_\rho(\boldsymbol{\chi})$ and the constraint violation $h(\boldsymbol{\chi})$ will reduce to zero. In other words, the equality constraint \eqref{constraint:P1.1_1} is satisfied. The detailed discussion of the convergence and optimality of the PDD framework can be found in \cite{shi2020penalty}.

\subsection{Proposed BCD Algorithm for Solving AL Problem \eqref{problem:AL_problem}} \label{sec:BCD_independent}
The key step of the proposed PDD-based algorithm is to solve the AL problem \eqref{problem:AL_problem}. In particular, we divide the set of the optimal variables $\boldsymbol{\chi}$ into two blocks, namely $\{\mathbf{U}, \mathbf{F}, \mathbf{P}, \mathbf{R}_s\}$ and $\{\boldsymbol{\theta}_t, \boldsymbol{\theta}_r\}$. Then, the BCD algorithm is invoked to solve each block iteratively while fixing the others, which leads to the following two subproblems.  

\subsubsection{Subproblem with respect to $\{\mathbf{U}, \mathbf{F}, \mathbf{P}, \mathbf{R}_s\}$} \label{sec:BCD_independent_PR}

The subproblem with respect to $\{\mathbf{U}, \mathbf{F}, \mathbf{P}, \mathbf{R}_s\}$ is givn by
\begin{subequations} \label{problem:P_R}
    \begin{align}
        \min_{\mathbf{U}, \mathbf{F}, \mathbf{P}, \mathbf{R}_s} &
        \mathrm{tr}(\mathbf{U}^{-1}) + \frac{1}{2\rho} \|\mathbf{F} - \mathbf{\Theta}_r \mathbf{G} \mathbf{R}_x \mathbf{G}^H \mathbf{\Theta}_r^H + \rho \mathbf{\Upsilon} \|^2\\
        \label{constraint:FIM}
        \mathrm{s.t.} \quad &\begin{bmatrix}
            \mathbf{J}_{\boldsymbol{\phi} \boldsymbol{\phi}}(\mathbf{F}) - \mathbf{U} &\mathbf{J}_{\boldsymbol{\phi} \tilde{\boldsymbol{\alpha}}}(\mathbf{F}) \\
            \mathbf{J}_{\boldsymbol{\phi} \tilde{\boldsymbol{\alpha}}}^T(\mathbf{F}) & \mathbf{J}_{\tilde{\boldsymbol{\alpha}} \tilde{\boldsymbol{\alpha}}}(\mathbf{F})
        \end{bmatrix} \succeq 0, \\
        \label{constraint:P_R_SINR}
        & \frac{1}{\overline{\gamma}_k} |\mathbf{u}_k^H \mathbf{p}_k|^2 \ge \sum_{i \in \mathcal{K} \backslash k} | \mathbf{u}_k^H \mathbf{p}_i|^2 + \mathbf{u}_k^H  \mathbf{R}_s \mathbf{u}_k + \sigma_k^2, \forall k, \\
        & \mathrm{tr}( \mathbf{P}\mathbf{P}^H + \mathbf{R}_s) \le P, \\
        & \mathbf{R}_s \succeq 0, \mathbf{U} \succeq 0,
    \end{align}
\end{subequations}
where $\mathbf{u}_k = \mathbf{G}^H \mathbf{\Theta}_t^H \mathbf{h}_k$ denotes the effective channel vector for user $k$. Constraint \eqref{constraint:P_R_SINR} is transformed from the communication SINR constraint \eqref{constraint:P1_2}. Although problem \eqref{problem:P_R} is apparently non-convex due to the constraint \eqref{constraint:P_R_SINR}, we now show that the global optimum of it can be obtained via semidefinite relaxation (SDR) approach.
\begin{lemma} \label{lemma:trasnmit}
    \emph{
        Given $\mathbf{P} = [\mathbf{p}_1,\dots,\mathbf{p}_K]^T$, there exists $\mathbf{R}_x$ and $\mathbf{R}_s$ that satisfy 
        $\mathbf{R}_x = \mathbf{P} \mathbf{P}^H + \mathbf{R}_s$ if and only if
        \begin{equation}
            \mathbf{R}_x \succeq \mathbf{P} \mathbf{P}^H = \sum_{k \in \mathcal{K}} \mathbf{p}_k \mathbf{p}_k^H
        \end{equation}
    }
\end{lemma}

\begin{proof}
    The necessity and sufficiency of this condition can be readily proved.
\end{proof}

By applying \textbf{Lemma \ref{lemma:trasnmit}} and defining the auxiliary variables $\mathbf{P}_k = \mathbf{p}_k \mathbf{p}_k^H, \forall k \in \mathcal{K},$ that satisfy $\mathbf{P}_k \succeq 0$ and $\mathrm{rank}(\mathbf{P}_k)=1$, we can formulate the SDR problem of \eqref{problem:P_R} as follows:
\begin{subequations} \label{problem:P_R_SDR}
    \begin{align}
        \min_{\scriptstyle \mathbf{U}, \mathbf{F}, \mathbf{R}_x, \atop \scriptstyle \{\mathbf{P}_k\}_{k \in \mathcal{K}}}  & 
        \mathrm{tr}(\mathbf{U}^{-1}) + \frac{1}{2\rho}\|\mathbf{F} - \mathbf{\Theta}_r \mathbf{G} \mathbf{R}_x \mathbf{G}^H \mathbf{\Theta}_r^H + \rho \mathbf{\Upsilon}\|^2\\
        \mathrm{s.t.} \quad & (1 + \frac{1}{\overline{\gamma}_k}) \mathbf{u}_k^H \mathbf{P}_k \mathbf{u}_k \ge \mathbf{u}_k^H \mathbf{R}_x \mathbf{u}_k + \sigma_k^2, \forall k,\\
        & \mathrm{tr}(\mathbf{R}_x) \le P, \\
        & \mathbf{R}_x \succeq \sum_{k \in \mathcal{K}} \mathbf{P}_k, \mathbf{R}_s \succeq 0, \mathbf{U} \succeq 0, \mathbf{P}_k \succeq 0, \forall k,
    \end{align}
\end{subequations}
where the non-convex rank-one constraints $\mathrm{rank}(\mathbf{P}_k)=1, \forall k \in \mathcal{K},$ are relaxed. The relaxed problem \eqref{problem:P_R_SDR}
is convex semidefinite programming (SDP), the global optimum of which can be efficiently obtained by the existing convex optimization solvers. However, due to the omission of the rank-one constraints, the global optimum of problem \eqref{problem:P_R_SDR} may have a higher rank, which may not be a solution to the original non-convex problem \eqref{problem:P_R}. Fortunately, in the following proposition, we show that the rank-one global optimum of problem \eqref{problem:P_R} can always be constructed from an arbitrary global optimum of problem \eqref{problem:P_R_SDR}.
\begin{proposition} \label{proposetion:SDR}
    \emph{Given an arbitrary global optimum $\tilde{\mathbf{R}}_x, \{\tilde{\mathbf{P}}_k\}_{k \in \mathcal{K}}$ of problem \eqref{problem:P_R_SDR}, the following solution is a global optimum of problem \eqref{problem:P_R}:
    \begin{equation}
        \mathbf{R}_x^\star = \tilde{\mathbf{R}}_x, \quad \mathbf{p}_k^\star = (\mathbf{u}_k^H \tilde{\mathbf{P}}_k \mathbf{u}_k)^{-1/2} \tilde{\mathbf{P}}_k \mathbf{u}_k.
    \end{equation} 
    }
\end{proposition}

\begin{proof}
    Please refer to \cite[Theorem 1]{liu2020joint}
\end{proof}

According to \textbf{Proposition \ref{proposetion:SDR}}, a global optimum of problem \eqref{problem:P_R} can be obtained by solving problem \eqref{problem:P_R_SDR}. Then, the optimal $\mathbf{R}_s$ can be calculated as 
\begin{equation}
    \mathbf{R}_s^\star = \mathbf{R}_x^\star - \sum_{k \in \mathcal{K}} \mathbf{p}_k^\star (\mathbf{p}_k^\star)^H.
\end{equation}

\subsubsection{Subproblem with respect to $\{\boldsymbol{\theta}_t, \boldsymbol{\theta}_r\}$}
The subproblem with respect to $\{\boldsymbol{\theta}_t, \boldsymbol{\theta}_r\}$ is given by 
\begin{subequations} \label{problem:theta}
    \begin{align}
        \min_{\boldsymbol{\theta}_t, \boldsymbol{\theta}_r} \quad & 
        \|\mathbf{F} - \mathbf{\Theta}_r \mathbf{G} \mathbf{R}_x \mathbf{G}^H \mathbf{\Theta}_r^H + \rho \mathbf{\Upsilon} \|^2 \\
        \mathrm{s.t.} \quad & \gamma_k \ge \overline{\gamma}_k, \forall k, \\
        & \beta_{t,n}^2  + \beta_{r,n}^2 \le 1, 0 \le \beta_{t,n}, \beta_{r,n} \le 1, \forall n.
    \end{align}
\end{subequations}
To facilitate the optimization of $\boldsymbol{\theta}_t$ and $\boldsymbol{\theta}_r$, we first transform the objective function and the communication SINR $\gamma_k$ into the more trackable forms.
In particular, we define $\bar{\mathbf{R}}_x = \mathbf{G} \mathbf{R}_x \mathbf{G}^H$. The eigenvalue decomposition of the matrix $\bar{\mathbf{R}}_x$ is given by 
\begin{equation} \label{eqn:eigen_decomposition}
    \bar{\mathbf{R}}_x = \sum_{j=1}^{R} \varrho_j \mathbf{v}_j \mathbf{v}_j^H = \sum_{j=1}^{R} \bar{\mathbf{v}}_j \bar{\mathbf{v}}_j^H,
\end{equation} 
where $\bar{\mathbf{v}}_j = \sqrt{\varrho_j} \mathbf{v}_j$ with $\varrho_j$ and $\mathbf{v}_j$ denoting the eigenvalue and the corresponding eigenvector, respectively,
and $R$ denotes the rank of the matrix $\bar{\mathbf{R}}_x$. As such, the objective function can be reformulated as 
\begin{align}
    & \big\|\mathbf{F} - \sum_{j=1}^{R} \mathbf{\Theta}_r \bar{\mathbf{v}}_j \bar{\mathbf{v}}_j^H \mathbf{\Theta}_r^H + \rho \mathbf{\Upsilon} \big\|^2 \nonumber \\
    = &\big\|\bar{\mathbf{F}} - \sum_{j=1}^{R} \mathrm{diag}(\bar{\mathbf{v}}_j) \boldsymbol{\theta}_r \boldsymbol{\theta}_r^H \mathrm{diag}(\bar{\mathbf{v}}_j)^H + \rho \mathbf{\Upsilon} \big\|^2.
\end{align}
Next, we define the following variables:
\begin{align}
    &\mathbf{\Phi}_{k,i} = \mathrm{diag}(\mathbf{h}_k^H) \mathbf{G} \mathbf{p}_i \mathbf{p}_i^H \mathbf{G}^H \mathrm{diag}(\mathbf{h}_k), \\
    &\mathbf{\Psi}_k = \mathrm{diag}(\mathbf{h}_k^H) \mathbf{G} \mathbf{R}_s \mathbf{G}^H \mathrm{diag}(\mathbf{h}_k), \\
    &\bar{\mathbf{V}}_j = \mathrm{diag}(\bar{\mathbf{v}}_j), \bar{\mathbf{F}} = \mathbf{F} + \rho \mathbf{\Upsilon}.
\end{align}
As a consequence, problem \eqref{problem:theta} can be rewritten as 
\begin{subequations} \label{problem:theta_transformed}
    \begin{align}
        \min_{\boldsymbol{\theta}_t, \boldsymbol{\theta}_r} \quad &
        \big\|\bar{\mathbf{F}} - \sum_{j=1}^{R} \bar{\mathbf{V}}_j \boldsymbol{\theta}_r \boldsymbol{\theta}_r^H \bar{\mathbf{V}}_j^H \big\|^2 \\
        \label{constraint:P1.7_1}
        \mathrm{s.t.} \quad & \frac{1}{\overline{\gamma}_k} \boldsymbol{\theta}_t^H \mathbf{\Phi}_{k,k}^* \boldsymbol{\theta}_t \ge \sum_{i \in \mathcal{K} \backslash k} \!\! \boldsymbol{\theta}_t^H \mathbf{\Phi}_{k,i}^* \boldsymbol{\theta}_t  + \boldsymbol{\theta}_t^H \mathbf{\Psi}_k^* \boldsymbol{\theta}_t + \sigma_k^2, \forall k, \\
        \label{constraint:P1.7_2}
        & |[\boldsymbol{\theta}_t]_n |^2  + |[\boldsymbol{\theta}_r]_n|^2 = 1,  \forall n.
    \end{align}
\end{subequations}
By observing that the objective function and all the constraints of problem \eqref{problem:theta_transformed} are in quadratic form with respect to the optimization variables, we also exploit the SDR approach to approximately solve it. By defining the auxiliary variables $\mathbf{Q}_i = \boldsymbol{\theta}_i \boldsymbol{\theta}_i^H, \forall i \in \{t,r\},$ that satisfies $\mathbf{Q}_i \succeq 0$ and $\mathrm{rank}(\mathbf{Q}_i) = 1$, the SDR problem of problem \eqref{problem:theta_transformed} is given by 
\begin{subequations} \label{problem:theta_SDR}
    \begin{align}
        \min_{\mathbf{Q}_t, \mathbf{Q}_r} \quad &
        \big\|\bar{\mathbf{F}} - \sum_{k=1}^{R} \tilde{\mathbf{V}}_k \mathbf{Q}_r \tilde{\mathbf{V}}_k^H \big\|^2 \\
        \mathrm{s.t.} \quad & \frac{1}{\overline{\gamma}_k}  \mathrm{tr}(\mathbf{\Phi}_{k,k}^* \mathbf{Q}_t) \ge 
        \sum_{i \in \mathcal{K} \backslash k} \mathrm{tr}(\mathbf{\Phi}_{k,i}^* \mathbf{Q}_t) \nonumber \\
        & \qquad \qquad \qquad \qquad + \mathrm{tr}(\mathbf{\Psi}_k^* \mathbf{Q}_t) + \sigma_k^2, \forall k, \\
        & [\mathbf{Q}_t]_{n,n} + [\mathbf{Q}_r]_{n,n} = 1,  \forall n, \\
        & \mathbf{Q}_t \succeq 0, \mathbf{Q}_r \succeq 0.
    \end{align}
\end{subequations}   
The above problem is also a convex SDP and thereby the global optimum of it can be efficiently obtained via the existing convex optimization solvers. While the SDR may result in the solution with the general rank, the eigenvalue decomposition or Gaussian randomization \cite{luo2010semidefinite} can be applied to construct a feasible rank-one solution to problem \eqref{problem:theta_transformed}. Note that a sufficient number of Gaussian randomization can achieve at least $\frac{\pi}{4}$-approximation of the optimal objective value of problem \eqref{problem:theta_transformed}. However, the slight performance loss caused by constructing the rank-one solution cannot theoretically guarantee the monotonicity of the objective value during the BCD iteration process. In this case, the penalty-based alternating minimization (AltMin) algorithm proposed in \cite{yu2021irs} with the provable convergence to the stationary point can be exploited, where the rank-one constraints are transformed as a penalty term in the objective function. Nevertheless, it is worth mentioning that the eigenvalue decomposition or Gaussian randomization can generally guarantees the convergence of the proposed algorithm in practice.

The overall BCD algorithm for solving problem \eqref{problem:AL_problem} is summarized in \textbf{Algorithm \ref{alg:BCD}}. Typically, when the AltMin algorithm is exploited for updating the block $\{\boldsymbol{\theta}_t, \boldsymbol{\theta_r}\}$, \textbf{Algorithm \ref{alg:BCD}} is guaranteed to converge to a stationary point of problem \eqref{problem:AL_problem} in polynomial time \cite{razaviyayn2013unified}. The main complexity of \textbf{Algorithm \ref{alg:BCD}} arises from solving problems \eqref{problem:P_R_SDR} and \eqref{problem:theta_SDR}. Given the solution accuracy $\epsilon$, the corresponding complexity via the interior-point method is in order of $\mathcal{O}((K^{6.5}M^{6.5} + N^{6.5})\log(1/\epsilon))$ and $\mathcal{O}((K+N)^{6.5}N^{6.5} \log(1/\epsilon))$, respectively \cite{liu2020joint, toh2008inexact}.
\begin{algorithm}[tb]
    \caption{BCD algorithm for solving problem \eqref{problem:AL_problem}.}
    \label{alg:BCD}
    \begin{algorithmic}[1]
        \STATE{Initialize feasible $\boldsymbol{\chi}$.}
        \REPEAT
        \STATE{update $\{\mathbf{U}, \mathbf{F}, \mathbf{P}, \mathbf{R}_s\}$ by solving problem \eqref{problem:P_R_SDR}.}
        \STATE{update $\{\boldsymbol{\theta}_t, \boldsymbol{\theta}_r\}$ by solving problem \eqref{problem:theta_SDR}.}
        \UNTIL{the fractional reduction of the objective value falls below a predefined threshold.}
    \end{algorithmic}
\end{algorithm}

\section{CRB Optimization Design with Coupled T\&R Phase-Shift} \label{sec:coupled}
In this section, we turn our attention to CRB optimization in the case of STARS with coupled T\&R phase shifts. Following a similar path in Section \ref{sec:coupled}, the PDD framework is also invoked. Regarding the coupled T\&R phase-shift constraints, a low-complexity iterative algorithm is proposed, where the amplitude and phase-shift coefficients of SATRS are updated alternately by the closed-form solutions.

\subsection{Problem Formulation}
According to \textbf{Proposition \ref{proposition:CRB_minimize}}, the optimization problem for minimizing $\mathrm{CRB}(\boldsymbol{\phi})$ with the coupled phase-shift constraints of STARS can be formulated as follows:
\begin{subequations} \label{problem:coupled_phase_shift}
    \begin{align}
        \tilde{\mathcal{P}}: \quad \min_{\mathbf{U}, \mathbf{P}, \mathbf{R}_s, \boldsymbol{\theta}_t, \boldsymbol{\theta}_r} \quad & 
        \mathrm{tr} \left( \mathbf{U}^{-1} \right), \\
        \label{constraint:coupled_phase}
        \mathrm{s.t.} \quad & \cos(\varphi_{t,n} - \varphi_{r,n}) = 0, \forall n,\\
        & \eqref{constraint:P1_1} - \eqref{constraint:P1_5},
    \end{align}
\end{subequations}
where the coupled phase-shift constraint \eqref{constraint:coupled_phase} can be reformulate as follows:
\begin{align}
    |\varphi_{t,n} - \varphi_{r,n}| = \frac{1}{2}\pi \text{ or } \frac{3}{2}\pi, \forall n
\end{align} 
Compared to the independent T\&R phase-shift model, the coupled T\&R phase-shift model imposed in \eqref{constraint:coupled_phase} makes the problem even more complex. In particular, the optimization subject to this constraint requires hybrid continuous and discrete control. For example, the phase shift $\varphi_{t,n}$ can be selected as any value in a continuous region $[0, 2 \pi]$, while the phase shift can only be selected from a discrete set $\{\varphi_{t,n} \pm \frac{1}{2} \pi, \varphi_{t,n} \pm \frac{3}{2} \pi\}$, which cannot be solved by existing methods. As a consequence, in the following subsections, we develop a new efficient PDD-based algorithm to solve the coupled T\&R phase-shift constraints.

\subsection{PDD Framework for Solving Problem $\tilde{\mathcal{P}}$}
In Section \ref{sec:indenependt}, we have proposed a PDD-based algorithm for the independent T\&R phase-shift model. In this case, we aim to obtain an equivalent form of problem \eqref{problem:coupled_phase_shift} where the coupled T\&R phased shifts are relaxed to the independent ones. Toward this idea, in addition to $\mathbf{F} = \mathbf{\Theta}_r \mathbf{G} \mathbf{R}_x \mathbf{G}^H \mathbf{\Theta}_r^H$, we define another set of auxiliary variables as follows:
\begin{align}
    \tilde{\boldsymbol{\theta}}_i = \boldsymbol{\theta}_i, \forall i \in \{t,r\},
\end{align}
where $\tilde{\boldsymbol{\theta}}_i = [\tilde{\beta}_{i,1} e^{j \tilde{\varphi}_{i,1}},\dots,\tilde{\beta}_{i,N} e^{j \tilde{\varphi}_{i,N}}]^T, \forall i \in \{t,r\}$. Then, problem $\tilde{\mathcal{P}}$ can be reformulated as
\begin{subequations} \label{problem:P_transform_2}
    \begin{align}
        \min_{\tilde{ \boldsymbol{\chi} }} \quad & \mathrm{tr}(\mathbf{U}^{-1})    \\
        \mathrm{s.t.} \quad 
        \label{constraint:phase_equality}
        & \tilde{\boldsymbol{\theta}}_r = \boldsymbol{\theta}_r, \tilde{\boldsymbol{\theta}}_t = \boldsymbol{\theta}_t, \\
        & \tilde{\beta}_{t,n}^2 + \tilde{\beta}_{r,n}^2 = 1, 0 \le \tilde{\beta}_{t,n}, \tilde{\beta}_{r,n} \le 1, \forall n, \\
        & |\tilde{\varphi}_{t,n} - \tilde{\varphi}_{r,n}| = \frac{1}{2}\pi \text{ or } \frac{3}{2}\pi, \forall n,\\
        \label{constraint:phase_relaxed}
        & \beta_{t,n}^2  + \beta_{r,n}^2 = 1, 0 \le \beta_{t,n}, \beta_{r,n} \le 1, \forall n, \\
        & \eqref{constraint:P1_2} - \eqref{constraint:P1_5}, \eqref{constraint:P1.1_1}, \eqref{constraint:P1.1_2}
    \end{align}
\end{subequations}
where $\tilde{ \boldsymbol{\chi} } \triangleq  \{\mathbf{U}, \mathbf{P}, \mathbf{R}_s, \mathbf{F}, \boldsymbol{\theta}_t, \boldsymbol{\theta}_r, \tilde{\boldsymbol{\theta}}_t, \tilde{\boldsymbol{\theta}}_r\}$ represents all optimization variables. Note that in problem \eqref{problem:P_transform_2}, the phase-shift constraints of the original optimization variables $\boldsymbol{\theta}_t$ and $\boldsymbol{\theta}_r$ have been relaxed to be independent, and the coupled T\&R phase-shift constraints are only related to the auxiliary variables $\tilde{\boldsymbol{\theta}}_t$ and $\tilde{\boldsymbol{\theta}}_r$. By introducing the Lagrangian dual variable $\boldsymbol{\lambda}_i \in \mathbb{C}^{N \times 1}, \forall i \in \{t,r\},$ for the additional equality constraints \eqref{constraint:phase_equality}, the following AL problem of \eqref{problem:P_transform_2} can be obtained:
\begin{subequations} \label{problem:AL_problem_2}
    \begin{align}
        \tilde{\mathcal{P}}_{\mathrm{AL}}(\rho, \mathbf{\Upsilon}, & \boldsymbol{\lambda}_i): \quad \min_{\tilde{ \boldsymbol{\chi} }} \quad 
        \mathrm{tr}(\mathbf{U}^{-1}) + \tilde{P}_\rho(\tilde{ \boldsymbol{\chi} }, \mathbf{\Upsilon}, \boldsymbol{\lambda}_i) \\
        \mathrm{s.t.} \quad & \eqref{constraint:P1_2} - \eqref{constraint:P1_5}, \eqref{constraint:P1.1_1}, \eqref{constraint:P1.1_2}, \eqref{constraint:phase_equality} - \eqref{constraint:phase_relaxed}, 
    \end{align}
\end{subequations}
where 
\begin{align}
    \tilde{P}_\rho(\tilde{ \boldsymbol{\chi} }, \mathbf{\Upsilon}, \boldsymbol{\lambda}_i) = &\frac{1}{2\rho} \big( \|\mathbf{F} - \mathbf{\Theta}_r \mathbf{G} \mathbf{R}_x \mathbf{G}^H \mathbf{\Theta}_r^H + \rho \mathbf{\Upsilon}\|^2 \nonumber \\
    &+ \sum_{ i \in \{t,r\} } \| \tilde{\boldsymbol{\theta}}_i - \boldsymbol{\theta}_i + \rho \boldsymbol{\lambda}_i \|^2  \big)
\end{align}
Thus, the PDD-based algorithm for solving problem $\tilde{\mathcal{P}}$ is summarized in \textbf{Algorithm \ref{alg:PDD_2}} and the constraint violation function is defined as 
\begin{align}
    \tilde{h}(\tilde{ \boldsymbol{\chi} }) = \max \big\{ &\|\mathbf{F} - \mathbf{\Theta}_r \mathbf{G} \mathbf{R}_x \mathbf{G}^H \mathbf{\Theta}_r^H \|_{\infty}, \nonumber \\
    &\| \tilde{\boldsymbol{\theta}}_t - \boldsymbol{\theta}_t \|_{\infty}, \| \tilde{\boldsymbol{\theta}}_r - \boldsymbol{\theta}_r \|_{\infty}
    \big\}.
\end{align} 

\begin{algorithm}[tb]
    \caption{PDD-based algorithm for solving problem $\tilde{\mathcal{P}}$.}
    \label{alg:PDD_2}
    \begin{algorithmic}[1]
        \STATE{Initialize feasible $\tilde{ \boldsymbol{\chi} }^{[0]}$, $\mathbf{\Upsilon}^{[0]}$, $\boldsymbol{\lambda}_i^{[0]}, \forall i \in \{t,r\}$, $\rho^{[0]} \ge 0$, and set $0 < c < 1$,  $n=1$.}
        \REPEAT
        \STATE{ $\tilde{ \boldsymbol{\chi} }^{[n+1]} = \mathrm{optimize} \left(  \tilde{\mathcal{P}}_{\mathrm{AL}}(\rho^{[n]}, \mathbf{\Upsilon}^{[n]}, \boldsymbol{\lambda}_i^{[n]}) \right) $ }
        \IF{$h(\tilde{ \boldsymbol{\chi} }^{[n+1]}) \le \eta^{[n]}$}
            \STATE{ $\mathbf{\Upsilon}^{[n+1]} = \mathbf{\Upsilon}^{[n]} + \frac{1}{\rho} ( \mathbf{F}^{[n+1]} - \mathbf{\Theta}_r^{[n+1]} \mathbf{G} \mathbf{R}_x^{[n+1]}$ $\times \mathbf{G}^H (\mathbf{\Theta}_r^{[n+1]})^H )$. }
            \STATE{ $\boldsymbol{\lambda}_i^{[n+1]} = \boldsymbol{\lambda}_i^{[n]} + \frac{1}{\rho} ( \tilde{\boldsymbol{\theta}}_i^{[n+1]} - \boldsymbol{\theta}_i^{[n+1]} ), \forall i \in \{t,r\} $. }
            \STATE{ $\rho^{[n+1]} = \rho^{[n]}$.}
        \ELSE
            \STATE{ $\mathbf{\Upsilon}^{[n+1]} = \mathbf{\Upsilon}^{[n]}$, $\boldsymbol{\lambda}_i^{[n+1]} = \boldsymbol{\lambda}_i^{[n]}, \forall i \in \{t,r\}$.}
            \STATE{ $\rho^{[n+1]} = c \rho^{[n]}$. }
        \ENDIF
        \STATE{$n=n+1$.}
        \UNTIL{the constraint violation $\tilde{h}(\tilde{ \boldsymbol{\chi} })$ falls below a predefined threshold.}
    \end{algorithmic}
\end{algorithm}

\subsection{Proposed BCD Algorithm for Solving AL Problem \eqref{problem:AL_problem_2}}
Similar to Section \ref{sec:BCD_independent}, the BCD is exploited to solve the AL problem \eqref{problem:AL_problem_2}, where the set of optimization variables $\tilde{ \boldsymbol{\chi} }$ is divided into three blocks, namely $\{\mathbf{U}, \mathbf{F}, \mathbf{P}, \mathbf{R}_s\}$, $\{\boldsymbol{\theta}_t, \boldsymbol{\theta}_r\}$, and $\{\tilde{\boldsymbol{\theta}}_t, \tilde{\boldsymbol{\theta}}_r\}$. The solutions of the corresponding subproblems are given as follows.

\subsubsection{Subproblem with respect to $\{\mathbf{U}, \mathbf{F}, \mathbf{P}, \mathbf{R}_s\}$} 
The subproblem with respect to $\{\mathbf{U}, \mathbf{F}, \mathbf{P}, \mathbf{R}_s\}$ for the coupled T\&R phase-shift model is the same as that for the independent T\&R phase-shift model. Thus, it can be solved following the same path in Section \ref{sec:BCD_independent_PR}. 

\subsubsection{Subproblem with respect to $\{\boldsymbol{\theta}_t, \boldsymbol{\theta}_r\}$}
The subproblem with respect to $\{\boldsymbol{\theta}_t, \boldsymbol{\theta}_r\}$ can be optimized with the relaxed independent T\&R phase-shift constraints, which is given by 
\begin{subequations} \label{problem:theta_2}
    \begin{align}
        \min_{\boldsymbol{\theta}_r, \boldsymbol{\theta}_t} \quad & 
        \|\mathbf{F} - \mathbf{\Theta}_r \mathbf{G} \mathbf{R}_x \mathbf{G}^H \mathbf{\Theta}_r^H + \rho \mathbf{\Upsilon}\|^2 \nonumber \\
        &+ \sum_{i \in \{t,r\}} \| \tilde{\boldsymbol{\theta}}_i - \boldsymbol{\theta}_i + \rho \boldsymbol{\lambda}_i \|^2 \\
        \mathrm{s.t.} \quad & \gamma_k \ge \overline{\gamma}_k, \forall k \\
        & \beta_{t,n}^2  + \beta_{r,n}^2 = 1, 0 \le \beta_{t,n}, \beta_{r,n} \le 1, \forall n.
    \end{align}
\end{subequations}
Then, by defining $\tilde{\boldsymbol{\upsilon}}_i = -(\tilde{\boldsymbol{\theta}}_i + \rho \boldsymbol{\lambda}_i), \forall i \in \{t,r\},$ and following a similar path in Section \ref{sec:BCD_independent}, the objective function can be reformulated as 
\begin{align}
    &\big\| \bar{\mathbf{F}} - \sum_{j=1}^{R} \bar{\mathbf{V}}_j \boldsymbol{\theta}_r \boldsymbol{\theta}_r^H \bar{\mathbf{V}}_j^H \big\|^2 + \sum_{i \in \{t,r\}} \| \tilde{\boldsymbol{\upsilon}}_i + \boldsymbol{\theta}_i \|^2 \nonumber \\
    = & \big\|\bar{\mathbf{F}} - \sum_{j=1}^{R} \hat{\mathbf{V}}_j \boldsymbol{\vartheta}_r \boldsymbol{\vartheta}_r^H \hat{\mathbf{V}}_j^H \big\|^2 + \sum_{i \in \{t,r\} } \left(\boldsymbol{\vartheta}_i^H \mathbf{\Xi}_i \boldsymbol{\vartheta}_i + \tilde{\boldsymbol{\upsilon}}_i^H \tilde{\boldsymbol{\upsilon}}_i \right),
\end{align}
where 
\begin{equation}
    \hat{\mathbf{V}}_j = \begin{bmatrix}
        \bar{\mathbf{V}}_j &\mathbf{0}_{N}
    \end{bmatrix},
    \mathbf{\Xi}_i = \begin{bmatrix}
        \mathbf{I}_{N} &\tilde{\boldsymbol{\upsilon}}_i \\
        \tilde{\boldsymbol{\upsilon}}_i^H & 0
    \end{bmatrix},
    \boldsymbol{\vartheta}_i = \begin{bmatrix}
        \kappa_i\boldsymbol{\theta}_i\\
        \kappa_i
    \end{bmatrix},
\end{equation}
and $|\kappa_i|^2=1$. Similarly, the minimum communication SINR can be reformulated as 
\begin{align}
    \frac{1}{\overline{\gamma}_k} \boldsymbol{\vartheta}_t^H \hat{\mathbf{\Phi}}_{k,k}^* \boldsymbol{\vartheta}_t \ge \sum_{i \in \mathcal{K} \backslash k} \boldsymbol{\vartheta}_t^H \hat{\mathbf{\Phi}}_{k,i}^* \boldsymbol{\vartheta}_t + \boldsymbol{\vartheta}_t^H \hat{\mathbf{\Psi}}_k^* \boldsymbol{\vartheta}_t + \sigma_k^2, \forall k,
\end{align}
where 
\begin{equation}
    \hat{\mathbf{\Phi}}_{k,i} = \begin{bmatrix}
        \mathbf{\Phi}_{k,i} & \mathbf{0}_{N} \\ 
        \mathbf{0}_{N}^T & 0
    \end{bmatrix}, \quad 
    \hat{\mathbf{\Psi}}_k = \begin{bmatrix}
        \mathbf{\Psi}_k & \mathbf{0}_{N} \\ 
        \mathbf{0}_{N}^T & 0
    \end{bmatrix}
\end{equation}
It can be observed that the objective function and all the constraints of problem \eqref{problem:theta_2} have been transformed into the homogeneous quadratic form. Thus, we also exploit the SDR approach to approximately solve it. In particular, by defining the auxiliary variables $\hat{\mathbf{Q}}_i = \boldsymbol{\vartheta}_i \boldsymbol{\vartheta}_i^H, \forall i \in \{t,r\}$, which satisfies $\hat{\mathbf{Q}}_i \succeq 0$ and $\mathrm{rank}(\hat{\mathbf{Q}}_i) = 1$, the SDR problem of \eqref{problem:theta_2} is given by 
\begin{subequations} \label{problem:theta_SDR_2}
    \begin{align}
        \min_{\hat{\mathbf{Q}}_t, \hat{\mathbf{Q}}_r} \quad &
        \big\| \bar{\mathbf{F}} - \sum_{j=1}^{R} \hat{\mathbf{V}}_j \hat{\mathbf{Q}}_t \hat{\mathbf{V}}_j^H \big\|^2 + \sum_{i \in \{t,r\} } \mathrm{tr}(\mathbf{\Xi}_i \hat{\mathbf{Q}}_i)   \\
        \mathrm{s.t.} \quad & \frac{1}{\overline{\gamma}_k}  \mathrm{tr}( \hat{\mathbf{\Phi}}_{k,k}^* \hat{\mathbf{Q}}_t) \ge 
        \sum_{i \in \mathcal{K} \backslash k} \mathrm{tr}( \hat{\mathbf{\Phi}}_{k,i}^* \hat{\mathbf{Q}}_t) \nonumber \\
        &\qquad \qquad \qquad \qquad + \mathrm{tr}( \hat{\mathbf{\Psi}}_k^* \hat{\mathbf{Q}}_t) + \sigma_k^2, \forall k, \\
        & [\hat{\mathbf{Q}}_t]_{n,n} + [\hat{\mathbf{Q}}_r]_{n,n} = 1,  \forall n, \\
        & [\hat{\mathbf{Q}}_t]_{N+1,N+1} = [\hat{\mathbf{Q}}_r]_{N+1,N+1} = 1, \\
        & \hat{\mathbf{Q}}_t \succeq 0, \hat{\mathbf{Q}}_r \succeq 0,
    \end{align}
\end{subequations}  
which is a convex SDP. Denote $\boldsymbol{\vartheta}_i^{\star}, \forall i \in \{t,r\},$ as the approximated rank-one solution. Then, the corresponding solution of problem \eqref{problem:theta_2} is given by 
\begin{equation}
    \boldsymbol{\theta}_i^{\star} = \frac{1}{ |[\boldsymbol{\vartheta}_i^{\star}]_{N+1}|} [\boldsymbol{\vartheta}_i^{\star}]_{1:N}, \forall i \in \{t,r\}.
\end{equation}
Similarly, the penalty-based AltMin method in \cite{yu2021irs} can be employed to guarantee the theoretical convergence of the BCD.

\subsubsection{Subproblem with respect to $\{\tilde{\boldsymbol{\theta}}_t, \tilde{\boldsymbol{\theta}}_r\}$}
The optimization variables $\{\tilde{\boldsymbol{\theta}}_t, \tilde{\boldsymbol{\theta}}_r\}$ only appear in the penalty term in the objective function as well as the coupled amplitude and phase-shift constraints. Thus, the subproblem with respect to $\{\tilde{\boldsymbol{\theta}}_t, \tilde{\boldsymbol{\theta}}_r\}$ is given by
\begin{subequations} \label{problem:tilde_theta}
    \begin{align}
        \min_{\tilde{\boldsymbol{\theta}}_t, \tilde{\boldsymbol{\theta}}_r} \quad & \sum_{i \in \{t,r\}}
        \| \tilde{\boldsymbol{\theta}}_i - \boldsymbol{\theta}_i + \rho \boldsymbol{\lambda}_i \|^2 \\
        \mathrm{s.t.} \quad  
        \label{constraint:tilde_theta_1}
        & \tilde{\beta}_{t,n}^2 + \tilde{\beta}_{r,n}^2 = 1, 0 \le \tilde{\beta}_{t,n}, \tilde{\beta}_{r,n} \le 1, \forall n, \\
        \label{constraint:tilde_theta_2}
        & |\tilde{\varphi}_{t,n} - \tilde{\varphi}_{r,n}| = \frac{1}{2}\pi \text{ or } \frac{3}{2}\pi, \forall n.
    \end{align}
\end{subequations}

\begin{figure*}[!t]
    \normalsize
    \setcounter{equation}{56}
    \begin{align}
        \label{eqn:break}
        \sum_{i \in \{t,r\}} \| \tilde{\boldsymbol{\theta}}_i - \boldsymbol{\theta}_i + \rho \boldsymbol{\lambda}_i \|^2 
        = &\sum_{i \in \{t,r\}} \| \mathrm{diag}(\tilde{\boldsymbol{\beta}}_i) \tilde{\mathbf{q}}_i + \boldsymbol{\upsilon}_i \|^2 \nonumber\\
        = &\sum_{i \in \{t,r\}} 2 \mathrm{Re}( \boldsymbol{\upsilon}_i^H \mathrm{diag}(\tilde{\boldsymbol{\beta}}_i) \tilde{\mathbf{q}}_i ) + \sum_{i \in \{t,r\}} \tilde{\mathbf{q}}_i^H \mathrm{diag}(\tilde{\boldsymbol{\beta}}_i) \mathrm{diag}(\tilde{\boldsymbol{\beta}}_i) \tilde{\mathbf{q}}_i + \sum_{i \in \{t,r\}} \boldsymbol{\upsilon}_i^H \boldsymbol{\upsilon}_i \nonumber\\
        \overset{(a)}{=} & \sum_{i \in \{t,r\}} 2 \mathrm{Re}( \boldsymbol{\upsilon}_i^H \mathrm{diag}(\tilde{\boldsymbol{\beta}}_i) \tilde{\mathbf{q}}_i ) + \sum_{i \in \{t,r\}} \sum_{n \in \mathcal{N}} \tilde{\beta}_{i,n}^2 + \sum_{i \in \{t,r\}} \boldsymbol{\upsilon}_i^H \boldsymbol{\upsilon}_i \nonumber \\
        \overset{(b)}{=} &\sum_{i \in \{t,r\}} 2 \mathrm{Re}( \boldsymbol{\upsilon}_i^H \mathrm{diag}(\tilde{\boldsymbol{\beta}}_i) \tilde{\mathbf{q}}_i ) + \underbrace{N + \sum_{i \in \{t,r\}} \boldsymbol{\upsilon}_i^H \boldsymbol{\upsilon}_i}_{\text{constant}}.
    \end{align}
    \hrulefill
    \vspace*{4pt}
\end{figure*}

\noindent In this problem, both constraints are non-convex and the second constraint even requires a binary decision, which is challenging to solve. However, in the following, we show that the amplitude coefficients and the phase shift coefficients can be optimized alternately using the closed-form solutions. To this end, we reformulate $\tilde{\boldsymbol{\theta}}_i$ as
\begin{equation}
    \tilde{\boldsymbol{\theta}}_i = \mathrm{diag}(\tilde{\boldsymbol{\beta}}_i) \tilde{\mathbf{q}}_i, \forall i \in \{t,r\},
\end{equation}
where $\tilde{\boldsymbol{\beta}}_i = [\tilde{\beta}_{i,1},\dots,\tilde{\beta}_{i,N}]^T$ and $\tilde{\mathbf{q}}_i = [e^{j \tilde{\varphi}_{i,1}},\dots,e^{j \tilde{\varphi}_{i,N}}]^T$. By defining $\boldsymbol{\upsilon}_i = -\boldsymbol{\theta}_i + \rho \boldsymbol{\lambda}_i, \forall i \in \{t,r\},$ the objective function of \eqref{problem:tilde_theta} can be reformulated into \eqref{eqn:break}, as shown at the top of the page. In \eqref{eqn:break}, the equality $(a)$ is due to the property $| [\tilde{\mathbf{q}}_i]_n |^2 = |e^{j \tilde{\varphi}_{i,n}}|^2 = 1$, and the equality $(b)$ stems from the property
$\tilde{\beta}_{t,n}^2 + \tilde{\beta}_{r,n}^2 = 1$. By removing the constant term in the objective function, problem \eqref{problem:tilde_theta} can be simplified as 
\begin{subequations}
    \begin{align}
        \min_{\tilde{\boldsymbol{\beta}}_t, \tilde{\boldsymbol{\beta}}_r, \tilde{\mathbf{q}}_t, \tilde{\mathbf{q}}_r} \quad &
        \sum_{i \in \{t,r\}} \mathrm{Re}( \boldsymbol{\upsilon}_i^H \mathrm{diag}(\tilde{\boldsymbol{\beta}}_i) \tilde{\mathbf{q}}_i ) \\
        \mathrm{s.t.} \quad & \eqref{constraint:tilde_theta_1}, \eqref{constraint:tilde_theta_2}.
    \end{align}
\end{subequations}
To solve it, we first give the following two propositions.

\begin{proposition} \label{proposition:optimal_phase}
    \emph{
    (\emph{Closed-form solution for coupled T\&R phase-shift}) 
    With the coupled T\&R phase-shift constraint \eqref{constraint:tilde_theta_2}, for any given $\tilde{\boldsymbol{\beta}}_t$ and $\tilde{\boldsymbol{\beta}}_r$, the optimal solution for the $n$-th entries $\tilde{q}_{t,n}$, $\tilde{q}_{r,n}$ of $\tilde{\mathbf{q}}_t$, $\tilde{\mathbf{q}}_r$ are chosen from the following two pairs of solutions 
    \begin{equation} \label{eqn:optimal_q}
        \begin{cases}
            \tilde{q}_{t,n} = e^{j (\pi - \angle \psi_n^+ )}, \tilde{q}_{r,n} = e^{j (\frac{3}{2}\pi - \angle \psi_n^+ )}, \\
            \tilde{q}_{t,n} = e^{j (\pi - \angle \psi_n^- )}, \tilde{q}_{r,n} = e^{j (\frac{1}{2}\pi - \angle \psi_n^- )}, 
        \end{cases}
    \end{equation}
    where $\psi_n^+ = \tilde{\upsilon}_{t,n}^* + j \tilde{\upsilon}_{r,n}^*$ and $\psi_n^- = \tilde{\upsilon}_{t,n}^* - j \tilde{\upsilon}_{r,n}^*$,
    such that the value of $\mathrm{Re}( \tilde{\upsilon}_{t,n}^* \tilde{q}_{t,n} ) + \mathrm{Re}( \tilde{\upsilon}_{r,n}^* \tilde{q}_{r,n} )$ is minimized. 
    Here, $\tilde{\upsilon}_{i,n}^*$ denotes the $n$-th entry of the vector $\tilde{\boldsymbol{\upsilon}}_i^H = \boldsymbol{\upsilon}_i^H \mathrm{diag}(\tilde{\boldsymbol{\beta}}_i), \forall i \in \{t,r\}$.
    }
\end{proposition}

\begin{proof}
    Please refer to Appendix~C.
\end{proof}

\begin{proposition} \label{proposition:optimal_amplitude}
    \emph{
    (\emph{Closed-form solution for amplitude})
    For any given $\tilde{\mathbf{q}}_t$ and $\tilde{\mathbf{q}}_r$, the optimal solution of the $n$-th entries $\tilde{\beta}_{t,n}$, $\tilde{\beta}_{r,n}$ of $\tilde{\boldsymbol{\beta}}_t$, $\tilde{\boldsymbol{\beta}}_r$ are given by
    \begin{align} 
        \label{eqn:optimal_beta}
        & \tilde{\beta}_{t,n} = \sin \omega_n, \tilde{\beta}_{r,n} = \cos \omega_n, \\
        \label{eqn:optimal_omega}
        & \omega_n = \begin{cases}
            -\frac{1}{2}\pi - \psi_n, &\text{if } \psi_n \in [-\pi, -\frac{1}{2}\pi),\\[-0.5em]
            0, & \text{if } \psi_n \in [-\frac{1}{2}\pi, \frac{1}{4}\pi), \\[-0.5em]
            \frac{1}{2}\pi, &\text{otherwise},
        \end{cases}
    \end{align}  
    where $\psi_n = \mathrm{sgn}(b_n) \arccos(\frac{a_n}{\sqrt{a_n^2 + b_n^2}}) \in [-\pi, \pi]$, $a_n = |\breve{\upsilon}_{t,n}^*| \cos (\angle \breve{\upsilon}_{t,n}^*)$, $b_n = |\breve{\upsilon}_{r,n}^*| \cos (\angle \breve{\upsilon}_{r,n}^*)$,
    and $\breve{\upsilon}_{i,n}^*$ is the $n$-th entry of the vector $\breve{\boldsymbol{\upsilon}}_i^H = \boldsymbol{\upsilon}_i^H\mathrm{diag}(\tilde{\mathbf{q}}_i), \forall i \in \{t,r\}$.
    }
\end{proposition}

\begin{proof}
    Please refer to Appendix~D.
\end{proof}

According to \textbf{Proposition \ref{proposition:optimal_phase}} and \textbf{Proposition \ref{proposition:optimal_amplitude}}, by fixing the amplitude (phase-shift) coefficients, the phase-shift (amplitude) coefficients have the optimal closed-form solutions. Consequently, problem \eqref{problem:tilde_theta} can also be solved iteratively, the detail of which is given in \textbf{Algorithm \ref{alg:closed_form_AO}}. Since the optimal solution is obtained at each step, the convergence of \textbf{Algorithm \ref{alg:closed_form_AO}} to the stationary points is guaranteed \cite{razaviyayn2013unified}.

\begin{algorithm}[tb]
    \caption{AO algorithm for solving problem \eqref{problem:tilde_theta}.}
    \label{alg:closed_form_AO}
    \begin{algorithmic}[1]
        \STATE{Initialize feasible $\tilde{\mathbf{q}}_t$, $\tilde{\mathbf{q}}_r$, $\tilde{\boldsymbol{\beta}}_t$, and $\tilde{\boldsymbol{\beta}}_r$.}
        \REPEAT
        \STATE{update each entry of $\tilde{\mathbf{q}}_t$ and $\tilde{\mathbf{q}}_r$ by \eqref{eqn:optimal_q}}
        \STATE{update each entry of $\tilde{\boldsymbol{\beta}}_t$ and $\tilde{\boldsymbol{\beta}}_r$ by \eqref{eqn:optimal_beta}.}
        \UNTIL{the fractional reduction of the objective value falls below a predefined threshold.}
    \end{algorithmic}
\end{algorithm}

\begin{algorithm}[tb]
    \caption{BCD algorithm for solving problem \eqref{problem:AL_problem_2}.}
    \label{alg:BCD_2}
    \begin{algorithmic}[1]
        \STATE{Initialize feasible $\tilde{ \boldsymbol{\chi} }$ such that $\mathbf{F} = \mathbf{\Theta}_r \mathbf{G} \mathbf{R}_x \mathbf{G}^H \mathbf{\Theta}_r^H$ and $\tilde{\boldsymbol{\theta}}_i = \boldsymbol{\theta}_i, \forall i \in \{t,r\}$.}
        \REPEAT
        \STATE{update $\{\mathbf{U}, \mathbf{F}, \mathbf{P}, \mathbf{R}_s\}$ by solving problem \eqref{problem:P_R_SDR}.}
        \STATE{update $\{\boldsymbol{\theta}_t, \boldsymbol{\theta}_r\}$ by solving problem \eqref{problem:theta_SDR_2}.}
        \STATE{update $\{\tilde{\boldsymbol{\theta}}_t, \tilde{\boldsymbol{\theta}}_r\}$ by solving problem \eqref{problem:tilde_theta} through Algorithm \ref{alg:closed_form_AO}.}
        \UNTIL{the fractional reduction of the objective value falls below a predefined threshold.}
    \end{algorithmic}
\end{algorithm}

\begin{table*}[tbp]
    \caption{System Parameters}
    \begin{center}
    \centering
    \resizebox{\textwidth}{!}{
        \begin{tabular}{|l|l|l||l|l|l|}
            \hline
            \centering
            $M$ & The number of antennas at the BS &$10$  & $N$ & The number of passive elements at the STARS &$6 \sim 18$\\
            \hline
            $N_s$ & The number of sensor elements &$2 \sim 24$  & $P$ & Transmit power at the BS &$30$ dBm\\
            \hline
            \centering
            $K$ & The number of communication users  & $4$ or $6$  & $\sigma_k^2, \sigma_s^2$ & The noise power  & $-110$ dBm\\
            \hline
            \centering
            $\alpha_{BR}, \alpha_{RU}$ &Path loss exponents &$2$ &$L$ &The length of a coherent time block &$100$ \\
            \hline
            \centering
            $\rho_0$&  The path loss at 1 m & $30$ dB &$\phi_h, \phi_v$ &Azimuth and elevation DOAs of the target &$120^\circ, 30^\circ$ \\
            \hline
            \centering
            $\varepsilon$&  Rician factor & $3$ dB & & & \\
            \hline
        \end{tabular}
    }
    \end{center}
    \label{table:parameters}
\end{table*}

The overall BCD algorithm for solving problem \eqref{problem:AL_problem_2} is summarized in \textbf{Algorithm \ref{alg:BCD_2}}. Similarly, when the AltMin method is adopted for optimizing the block $\{\boldsymbol{\theta}_t, \boldsymbol{\theta}_r\}$, the convergence of \textbf{Algorithm \ref{alg:BCD_2}} to the stationary point can be theoretically guaranteed \cite{razaviyayn2013unified}. The complexity of this algorithm is summarized as follows. Firstly, the complexities of using the interior-point method to solve problem \eqref{problem:P_R_SDR} and \eqref{problem:theta_SDR_2} are in order of $\mathcal{O}((K^{6.5}M^{6.5} + N^{6.5})\log(1/\epsilon))$ and $\mathcal{O}((K+N+2)^{6.5}(N+1)^{6.5} \log(1/\epsilon))$, respectively \cite{liu2020joint,toh2008inexact}. Moreover, in \textbf{Algorithm \ref{alg:closed_form_AO}}, the complexities of updating $\{\tilde{\mathbf{q}}_t, \tilde{\mathbf{q}}_r\}$ and $\{\tilde{\boldsymbol{\beta}}_t, \tilde{\boldsymbol{\beta}}_r\}$ are in order of $\mathcal{O}(4N)$ and $\mathcal{O}(2N)$, respectively.

\section{Numerical Results} \label{sec:results}

In this section, the numerical results obtained through Monte Carlo simulations are provided to evaluate the performance of the proposed STARS-enabled ISAC system. We assume that the BS is $40$ m away from the STARS. The communication users are randomly distributed within $20 \sim 50$ m from the transmission side of STARS, and the sensing target is $30$ m away from the reflection side of the STARS. We assume the Rician channel model for all communication channels. Thus, the channels $\mathbf{G}$ and $\mathbf{h}_k, \forall k \in \mathcal{K},$ are given by  
\begin{align}
    \mathbf{G} &= \sqrt{ \frac{\rho_0}{d_{BR}^{\alpha_{BR}} } }   \left( \sqrt{\frac{\varepsilon}{1 + \varepsilon}} \mathbf{G}^{\mathrm{LoS}} + \sqrt{\frac{1}{1 + \varepsilon}} \mathbf{G}^{\mathrm{NLoS}} \right), \\
    \mathbf{h}_k &= \sqrt{ \frac{\rho_0}{d_{RU,k}^{\alpha_{RU}} } } \left( \sqrt{\frac{\varepsilon}{1 + \varepsilon}} \mathbf{h}_k^{\mathrm{LoS}} + \sqrt{\frac{1}{1 + \varepsilon}} \mathbf{h}_k^{\mathrm{NLoS}} \right), 
\end{align} 
where $d_{BR}$ and $d_{RU,k}$ denote the BS-STARS distance and the STARS-user-$k$ distance, respectively, $\alpha_{BR}$ and $\alpha_{RU}$ denotes the corresponding path loss exponents, $\rho_0$ denotes the path loss at the reference distance of $1$ m, $\varepsilon$ denotes the Rician factor, $\mathbf{G}^{\mathrm{LoS}}$ and $\mathbf{h}_k^{\mathrm{LoS}}$ are the deterministic LoS component, and $\mathbf{G}^{\mathrm{NLoS}}$ and $\mathbf{h}_k^{\mathrm{NLoS}}$ are the random non-LoS component modeled as Rayleigh fading. The main
adopted system parameters are given in Table \ref{table:parameters}. The CVX toolbox \cite{cvx} is used to solve all convex problems involved in the proposed algorithms. The convergence threshold of the PDD-based algorithms and the BCD-based algorithms are set as $10^{-4}$ and $10^{-3}$, respectively. Without loss of generality, we assume that all communication users have the same SINR requirement, i.e., $\overline{\gamma}_k = \overline{\gamma}, \forall k \in \mathcal{K}$.  

To verify the efficiency of the proposed framework, we compared it with a baseline that employs one conventional reflecting-only RIS and one conventional transmitting-only RIS, both with $N/2$ elements and adjacent to each other at the same location as the STARS. This baseline is essentially a special case of STARS where the amplitudes of the TCs and RCs are fixed to $\boldsymbol{\beta}_t = [\mathbf{1}_{1 \times N/2}, \mathbf{0}_{1 \times N/2}]^T$ and $\boldsymbol{\beta}_r = [\mathbf{0}_{1 \times N/2}, \mathbf{1}_{1 \times N/2}]^T$. Therefore, the resultant optimization problem can also be solved by the proposed PDD-based algorithm. All following numerical results are obtained by averaging over $50$ random channel realizations unless otherwise specified.

\begin{figure}[t!]
    \centering
    \includegraphics[width=0.4\textwidth]{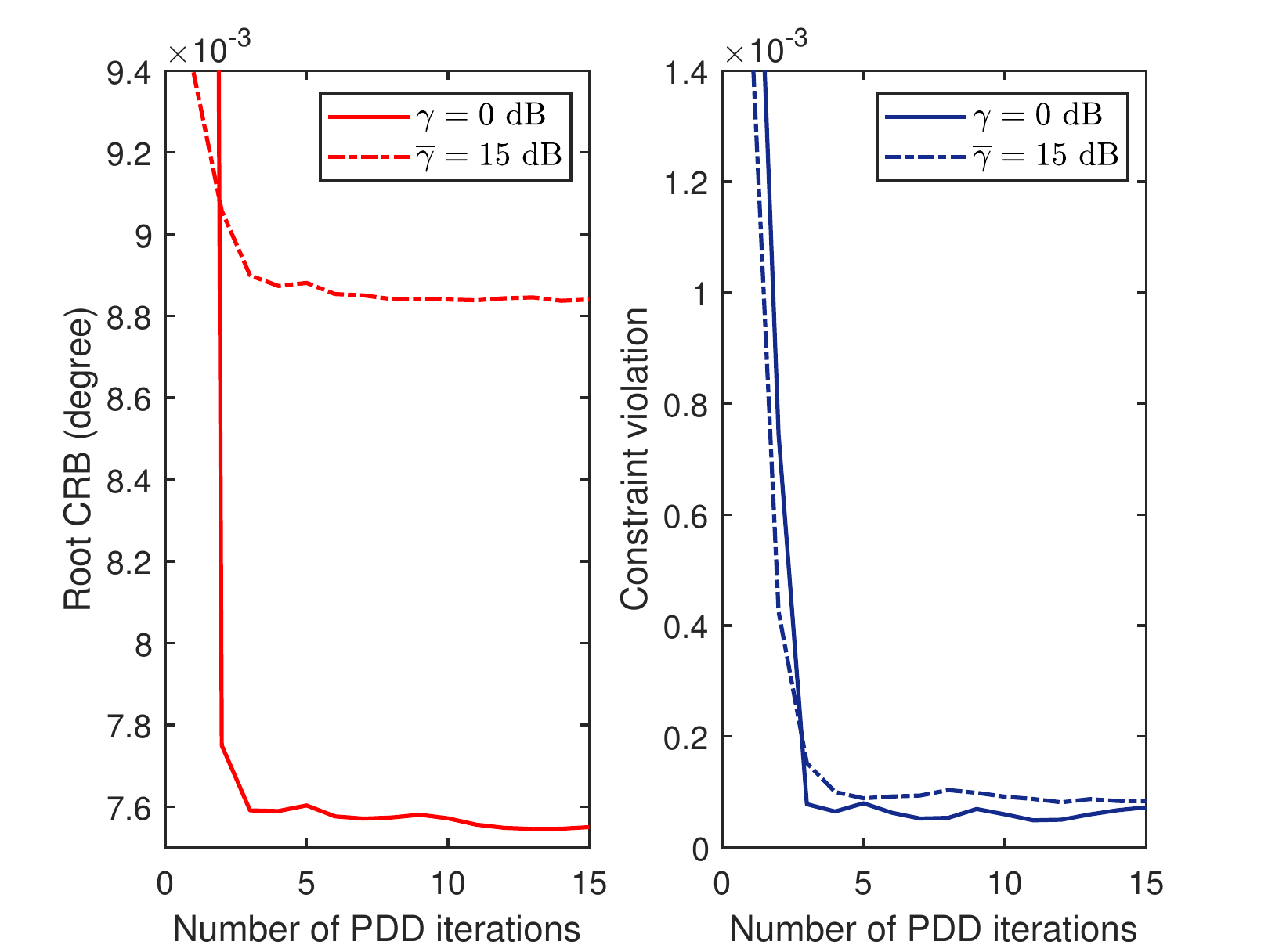}
    \caption{The convergence behavior of Algorithm \ref{alg:PDD}.}
    \label{fig:PDD_1}
\end{figure}
\begin{figure}[t!]
    \centering
    \includegraphics[width=0.4\textwidth]{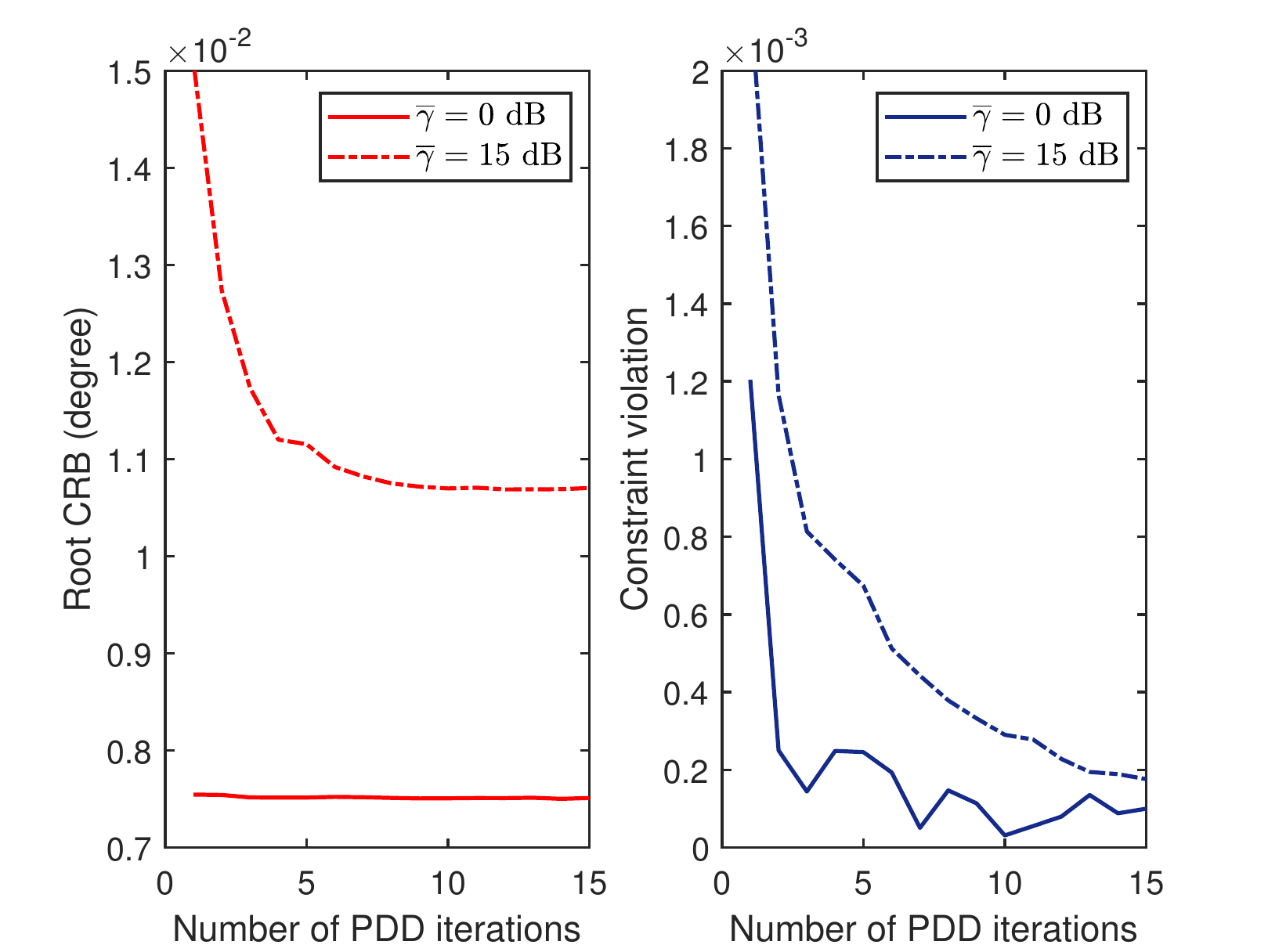}
    \caption{The convergence behavior of Algorithm \ref{alg:PDD_2}.}
    \label{fig:PDD_2}
\end{figure}

\subsection{Convergence Performance of the Proposed Algorithms}
In Fig. \ref{fig:PDD_1} and Fig. \ref{fig:PDD_2}, we examine the convergence performance of the proposed algorithms over a random channel realization. In particular, the initialization point of \textbf{Algorithm \ref{alg:PDD}} is randomly selected while that of \textbf{Algorithm \ref{alg:PDD_2}} is selected as the output of \textbf{Algorithm \ref{alg:PDD}}. Note that the original unit of the root CRB is radians, which is difficult to intuitively match the accuracy of actual DOA estimates. Therefore, we convert the unit of the root CRB from radians to degrees. In terms of the root CRB and the constraint violation, it can be observed that by using \textbf{Algorithm \ref{alg:PDD}} and \textbf{Algorithm \ref{alg:PDD_2}}, they can converge well for both independent and coupled phase-shift models. Moreover, when the SINR threshold is low, namely $\overline{\gamma}=0$dB, the \textbf{Algorithm \ref{alg:PDD_2}} almost converges at the first PDD iteration. This is because, with the low SINR threshold, the coupled phase-shift model is almost identical to the independent phase-shift model. Consequently, when \textbf{Algorithm \ref{alg:PDD_2}} is initialized as the output of \textbf{Algorithm \ref{alg:PDD}}, the optimal solution is almost reached at the beginning. This phenomenon will be further explained in the following numerical results.

\subsection{Root CRB Versus Communication SINR Threshold}
In Fig. \ref{fig:CRB_vs_SINR}, we studied the achieved root CRB versus the communication SINR threshold considering different user numbers. As can be observed, there is a tradeoff between the sensing and communication performance. This is because the higher communication performance requires more resources such as power and DoFs in the communication space, resulting in lower sensing performance. It can also be seen that independent of the phase shift model, STARS always outperforms conventional RIS in achieving a lower CRB. The conventional RIS almost becomes infeasible when $K=6$ and $\overline{\gamma} > 5$dB. This is indeed expected since conventional RIS utilizes a limited number of elements for sensing and communication spaces and therefore cannot achieve the same DoFs as STARS. Furthermore, the performance gap between coupled and independent phase shifts of the STARS is small when the communication SINR threshold $\overline{\gamma}$ is weak. However, as $\overline{\gamma}$ increases, the transmission phase shifts for the communication space become more stringent, which also limits the reflection phase shifts for the sensing space when the phase shifts are coupled, leading to a larger performance gap with the independent phase shifts.

\begin{figure}[t!]
    \centering
    \includegraphics[width=0.4\textwidth]{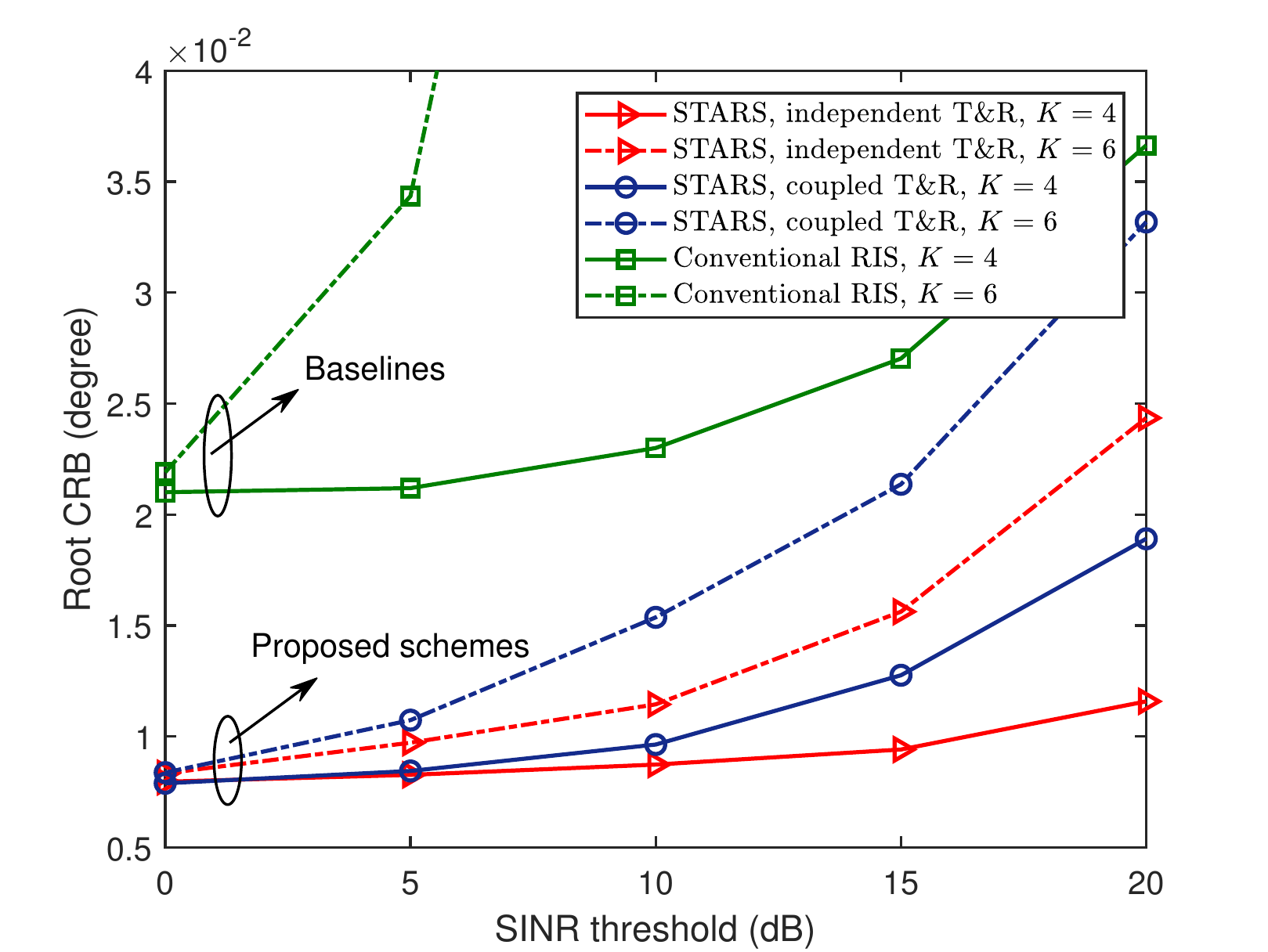}
    \caption{Root CRB versus the communication SINR threshold $\overline{\gamma}$ for $N=10$ and $N_s =5$.}
    \label{fig:CRB_vs_SINR}
\end{figure}
\begin{figure}[t!]
    \centering
    \includegraphics[width=0.4\textwidth]{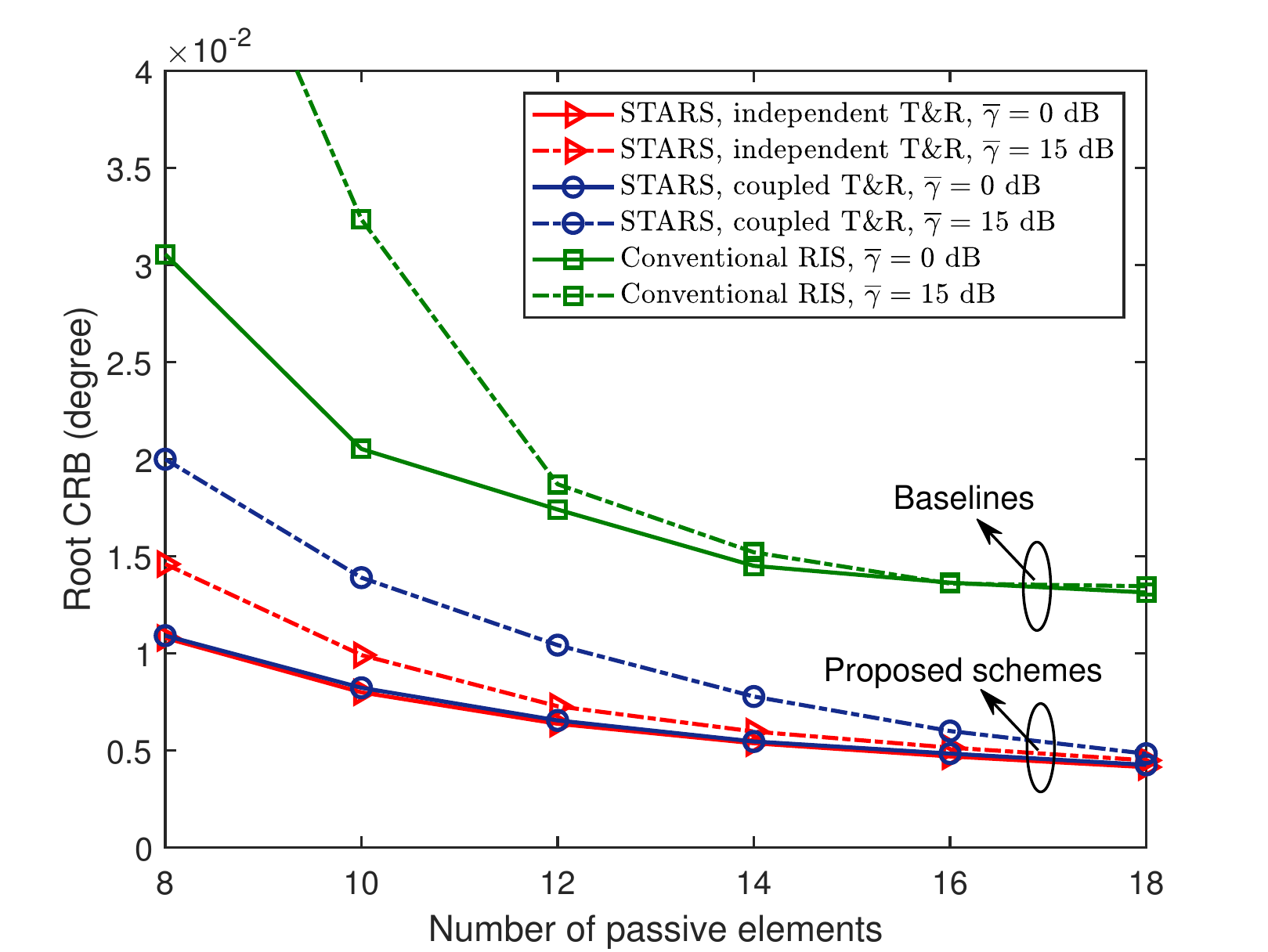}
    \caption{Root CRB versus the number of passive elements $N$ for $N_s = 5$ and $K=4$.}
    \label{fig:passive}
\end{figure}

\subsection{Root CRB Versus Number of Passive Elements}
In Fig. \ref{fig:passive}, we illustrate the impact of the number of passive elements $N$ of STARS when $N_s = 5$ and $K=4$. It can be observed that as the number of passive elements increases, the root CRB decreases and the communication SINR threshold $\overline{\gamma}$ has less influence on the CRB performance. This is because more passive elements introduce more DoFs to construct a more directional sensing beam and achieve a more flexible communication channel tuning capability. Due to the same reason, when $\overline{\gamma}=15$dB, the root CRB achieved by the two phase-shift models finally converges to the same value as $N$ increases. However, when $\overline{\gamma}=0$dB, the performance of the coupled phase-shift model is always comparable to that of the independent phase-shift model.

\begin{figure}[t!]
    \centering
    \includegraphics[width=0.4\textwidth]{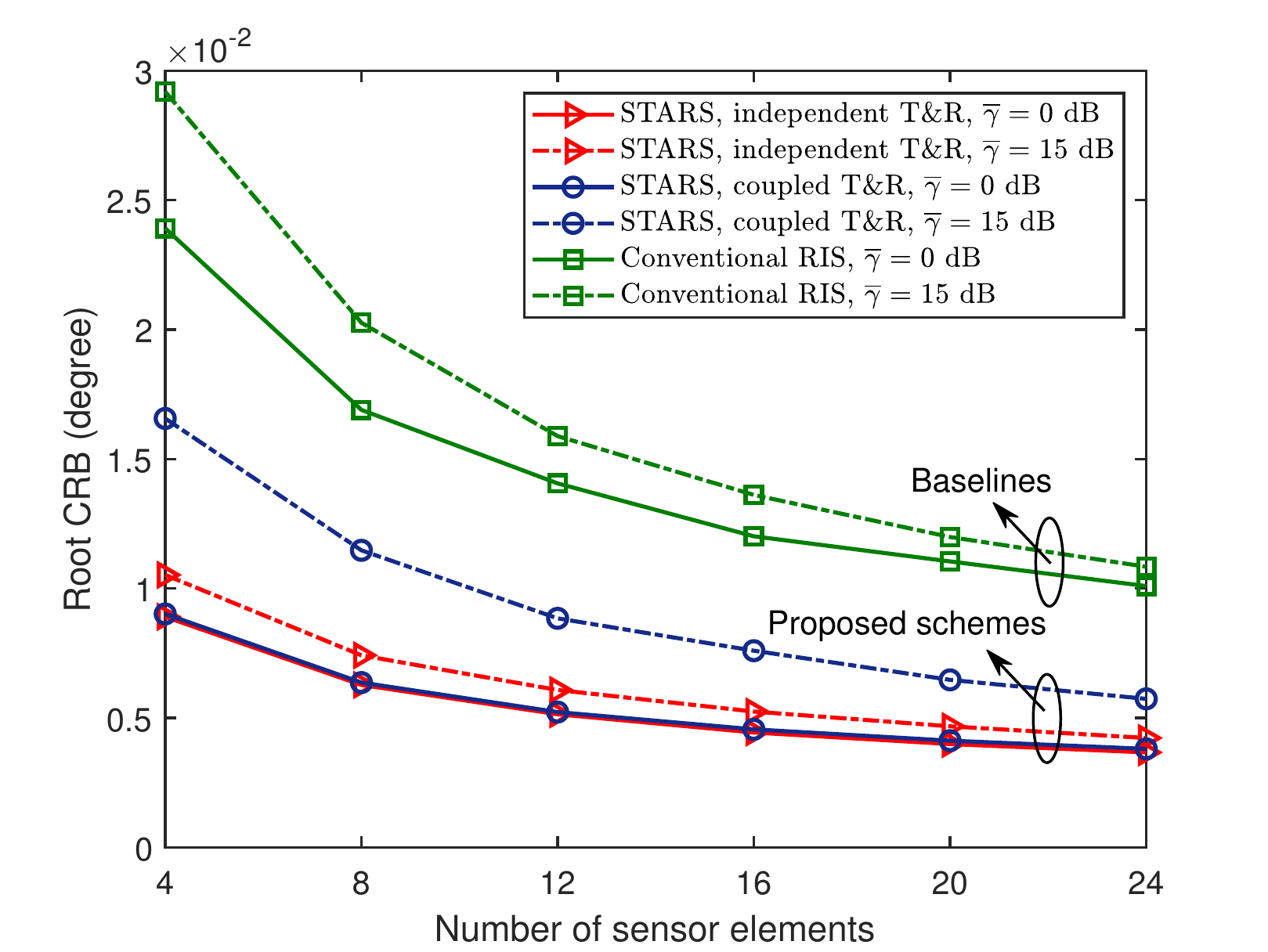}
    \caption{Root CRB versus the number of sensor elements $N_s$ for $N = 10$ and $K=4$.}
    \label{fig:sensor}
\end{figure}
\begin{figure}[t!]
    \centering
    \includegraphics[width=0.4\textwidth]{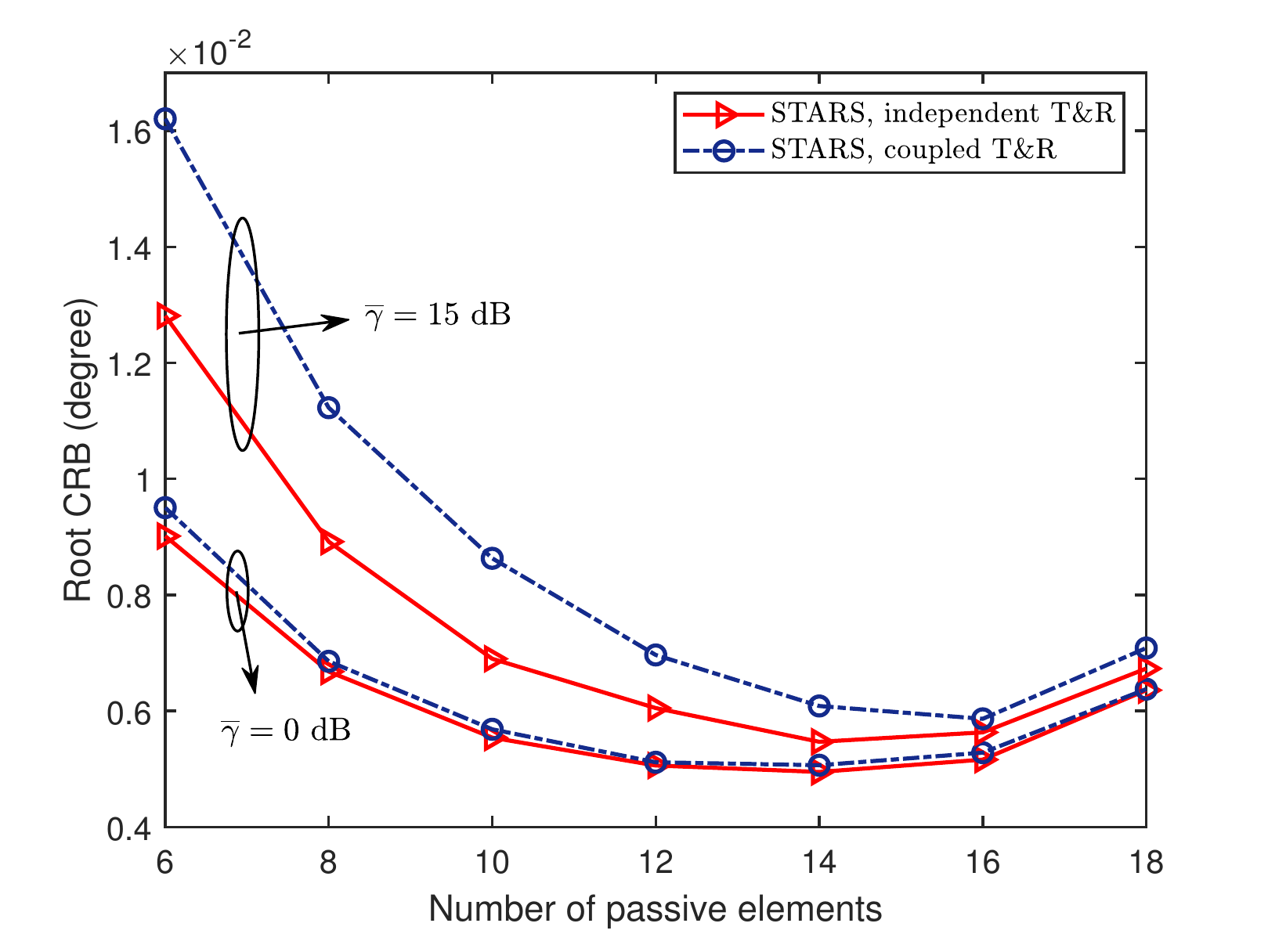}
    \caption{Tradeoff between the number of passive and sensor elements for $N + N_s = 20$ and $K=4$.}
    \label{fig:element_tradeoff}
\end{figure}

\begin{figure*}[t!]
    \centering
    \captionsetup{font={footnotesize}}
    \begin{subfigure}[t]{0.24\textwidth}
        \centering
        \includegraphics[width=1\textwidth]{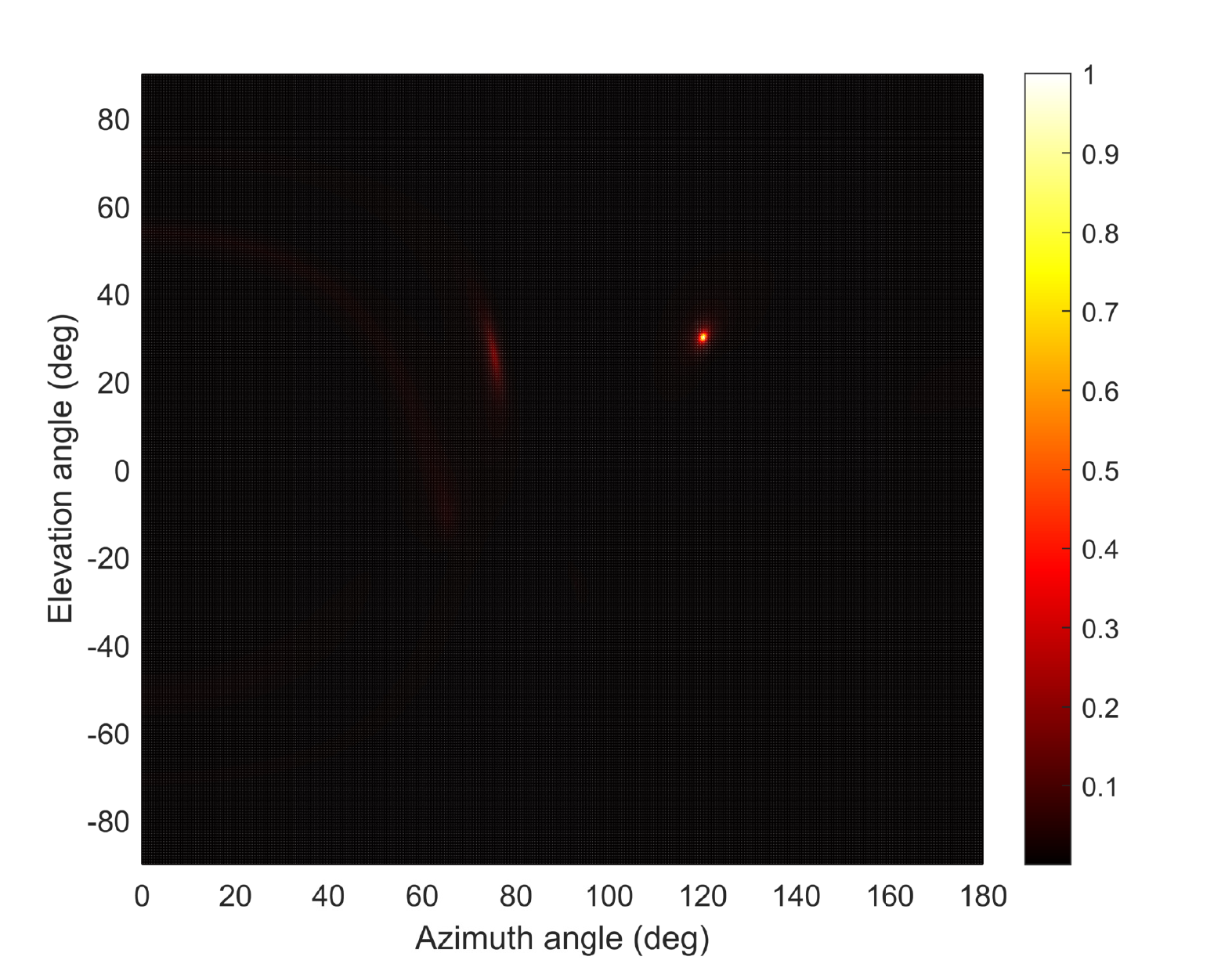}
        \caption{STARS, independent, $\overline{\gamma} = 0$dB.}
    \end{subfigure}
    \begin{subfigure}[t]{0.24\textwidth}
        \centering
        \includegraphics[width=1\textwidth]{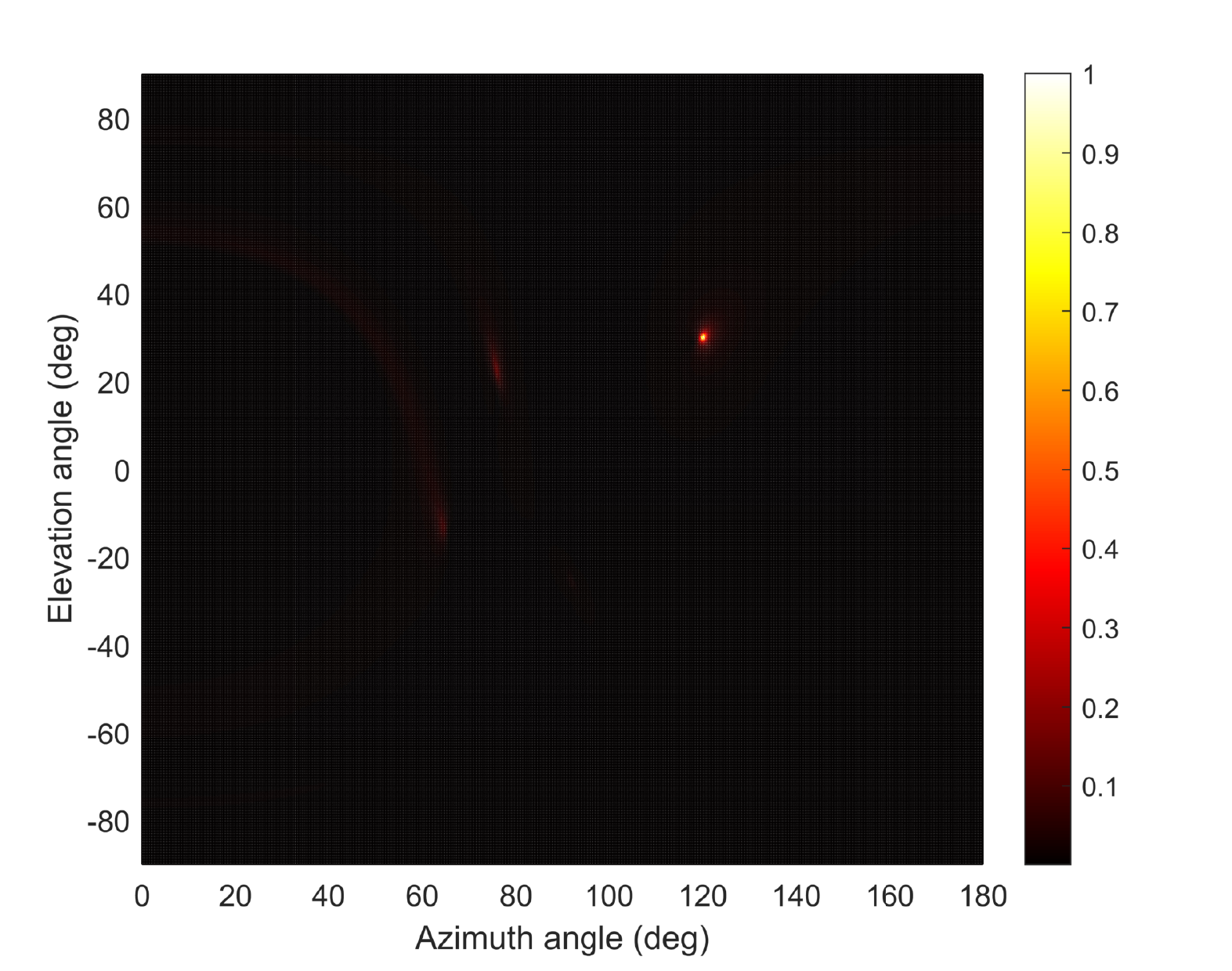}
        \caption{STARS, coupled, $\overline{\gamma} = 0$dB.}
    \end{subfigure}
    \begin{subfigure}[t]{0.24\textwidth}
        \centering
        \includegraphics[width=1\textwidth]{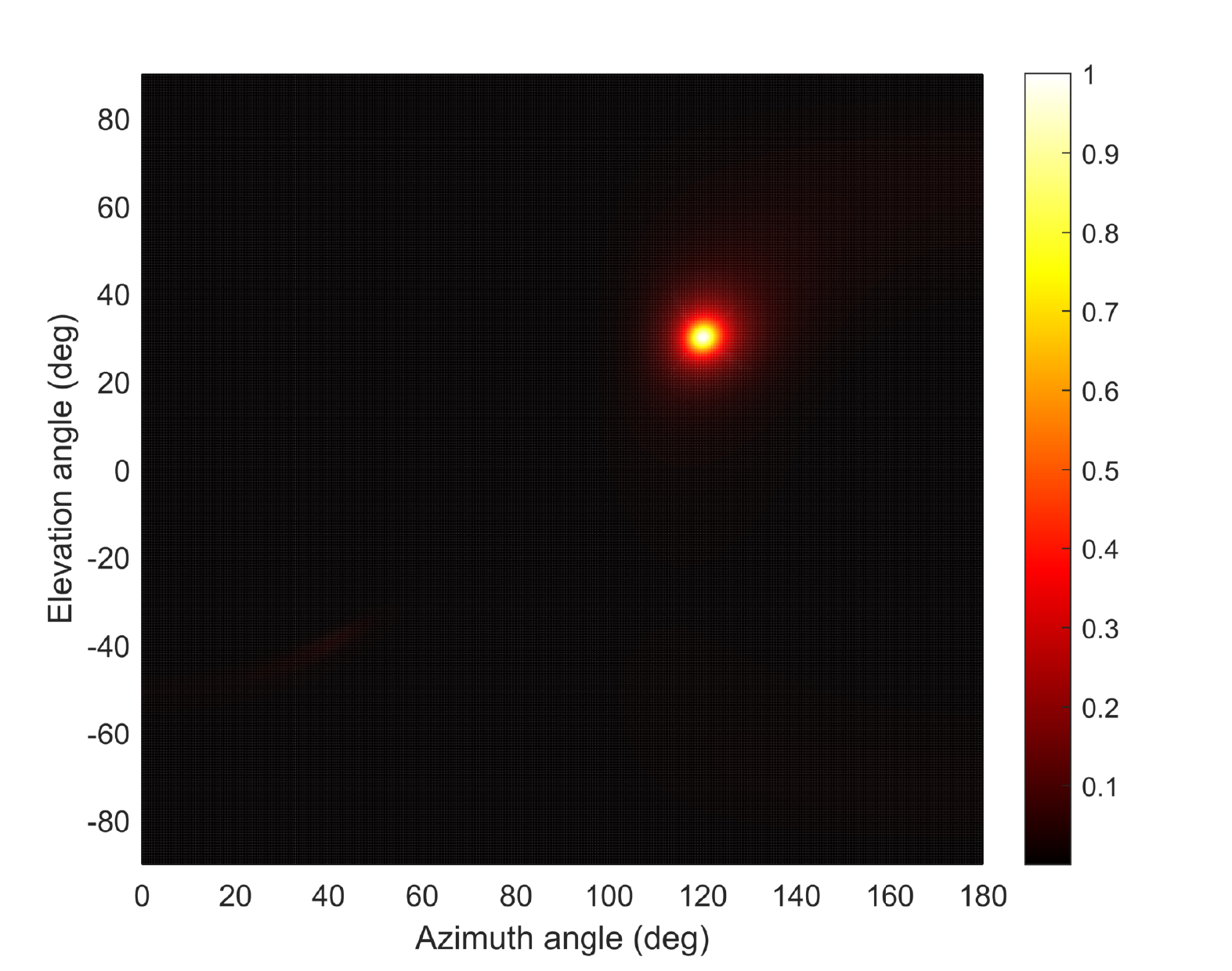}
        \caption{Conventional RIS, $\overline{\gamma} = 0$dB.}
    \end{subfigure}
    \begin{subfigure}[t]{0.24\textwidth}
        \centering
        \includegraphics[width=1\textwidth]{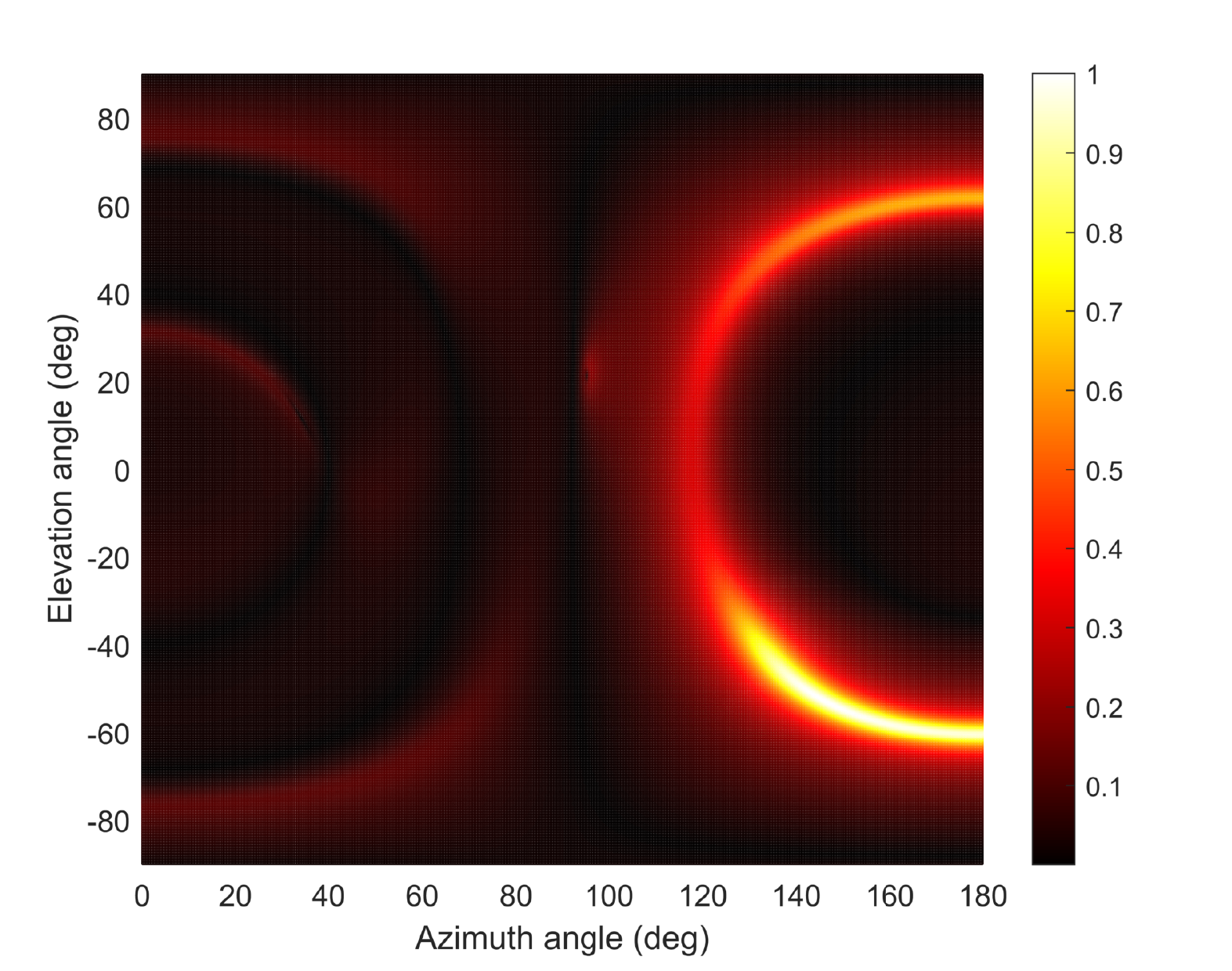}
        \caption{STARS, random, $\overline{\gamma} = 0$dB.}
    \end{subfigure}
    \begin{subfigure}[t]{0.24\textwidth}
        \centering
        \includegraphics[width=1\textwidth]{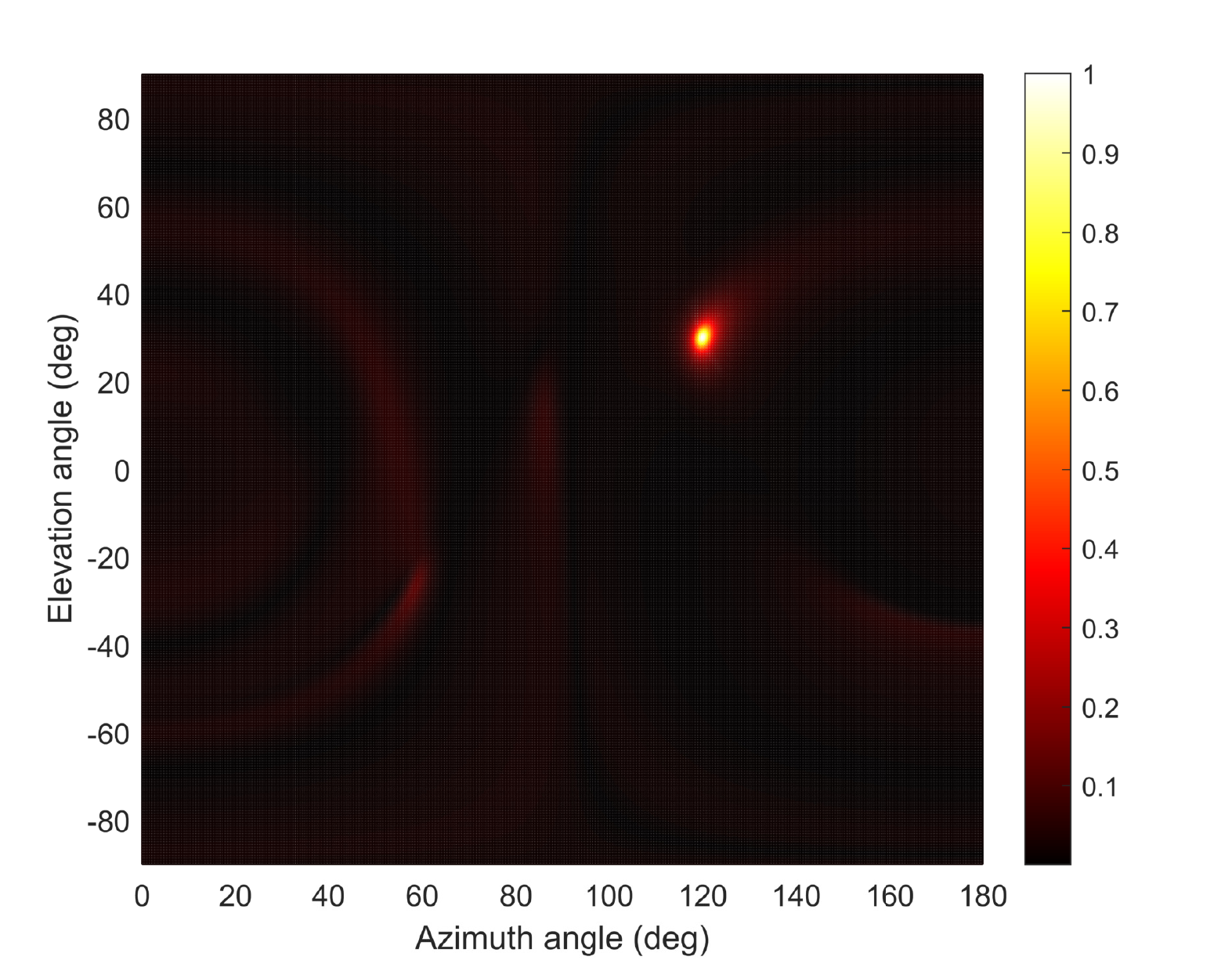}
        \caption{STARS, independent, $\overline{\gamma} = 20$dB.}
    \end{subfigure}
    \begin{subfigure}[t]{0.24\textwidth}
        \centering
        \includegraphics[width=1\textwidth]{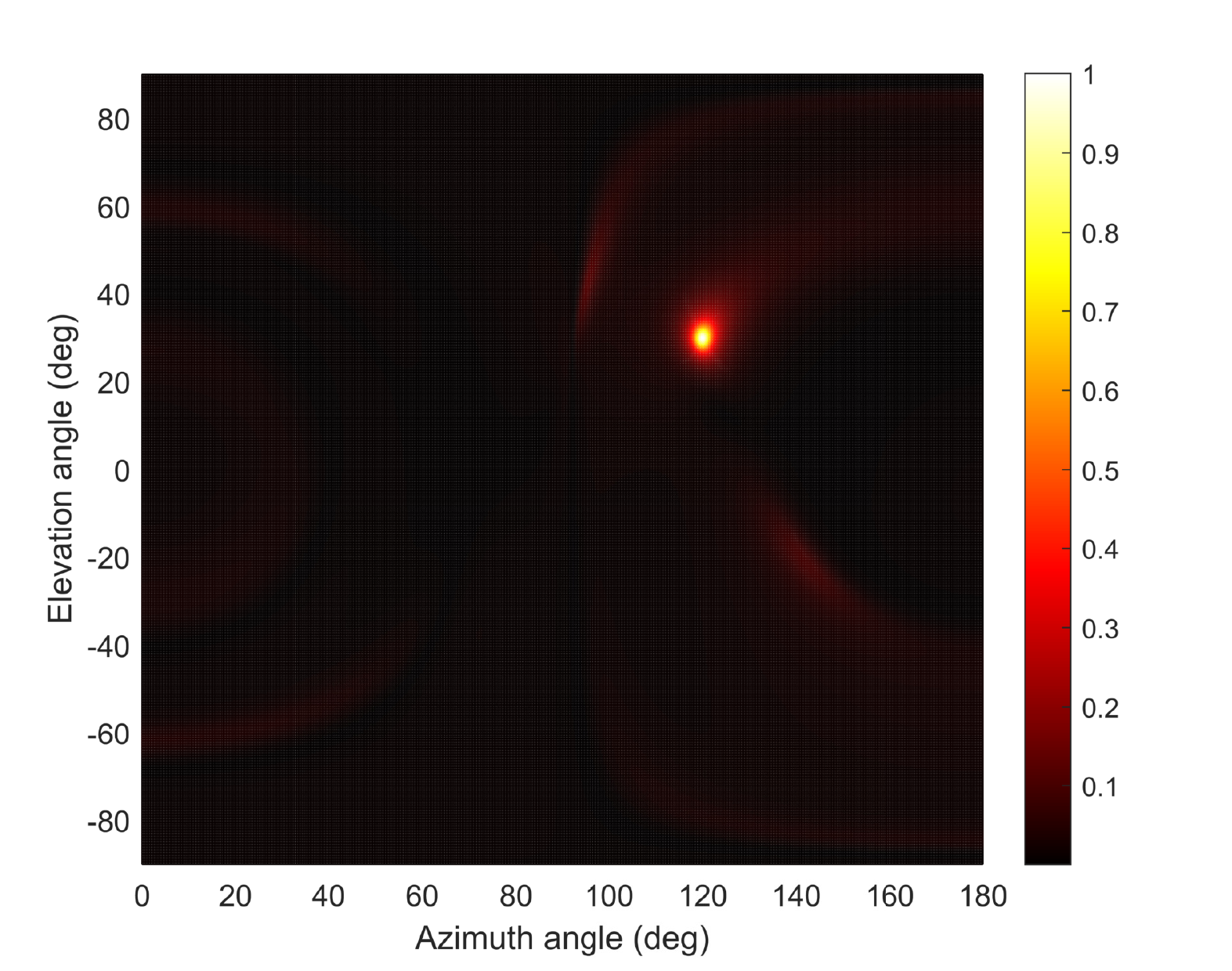}
        \caption{STARS, coupled, $\overline{\gamma} = 20$dB.}
    \end{subfigure}
    \begin{subfigure}[t]{0.24\textwidth}
        \centering
        \includegraphics[width=1\textwidth]{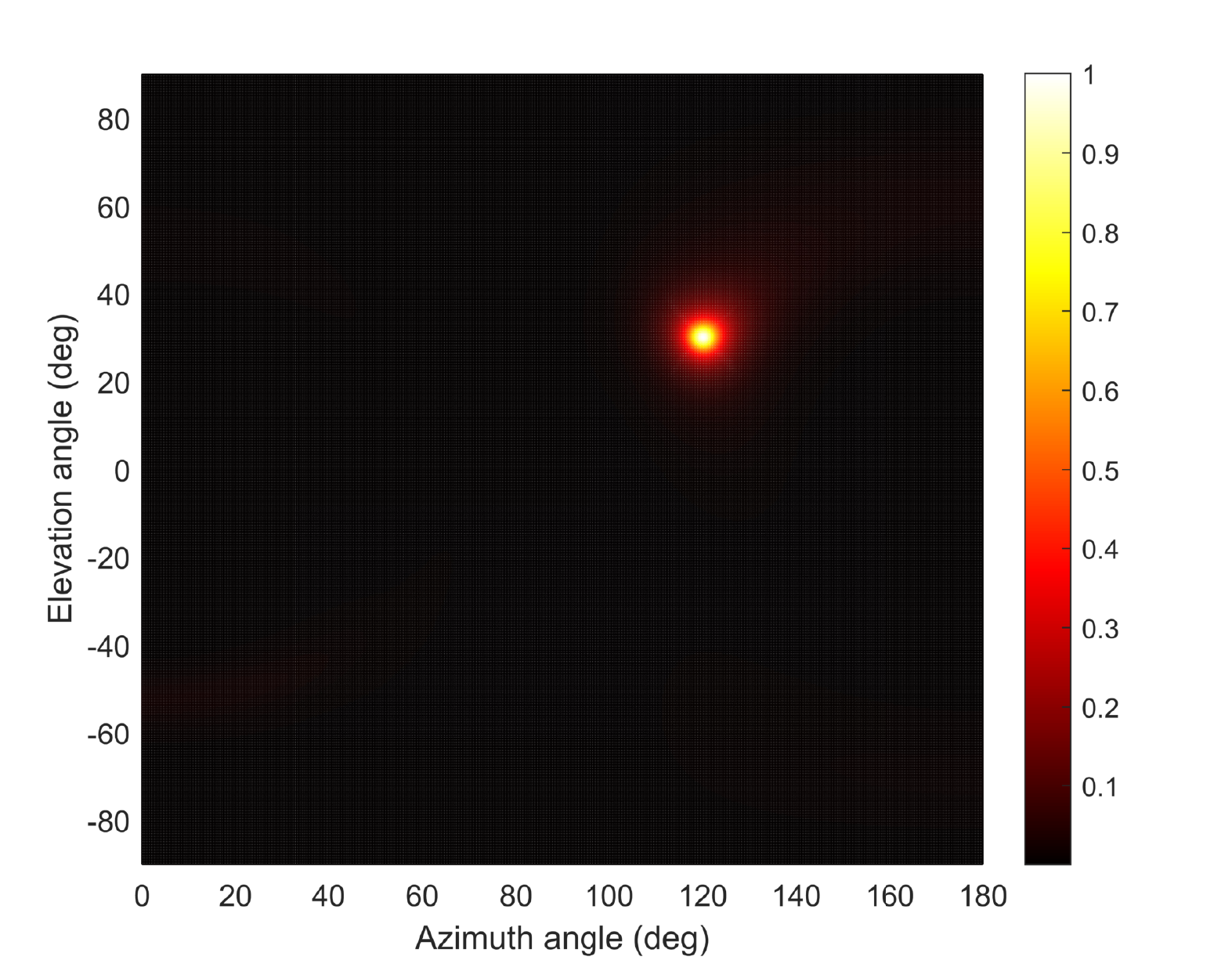}
        \caption{Conventional RIS, $\overline{\gamma} = 20$dB.}
    \end{subfigure}
    \begin{subfigure}[t]{0.24\textwidth}
        \centering
        \includegraphics[width=1\textwidth]{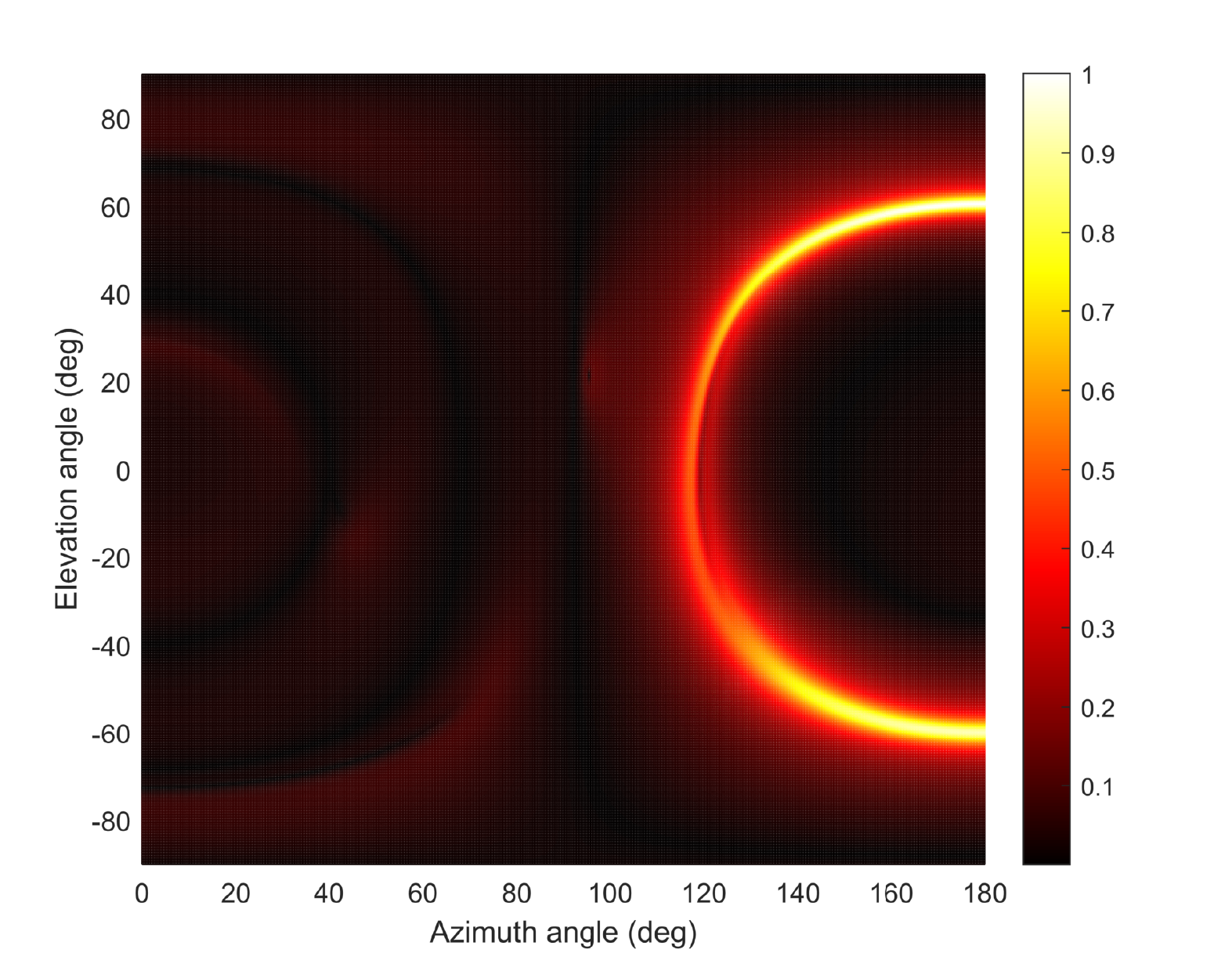}
        \caption{STARS, random, $\overline{\gamma} = 20$dB.}
    \end{subfigure}
    \caption{\small Normalized spectrum of MLE by different schemes for $N = 10$, $N_s = 5$, and $K=4$.}
    \label{fig:MLE}
\end{figure*}

\subsection{Root CRB Versus Number of Sensor Elements}
In Fig. \ref{fig:sensor}, we further studied the impact of the number of active sensor elements $N_s$ when $N=10$. We can see that all schemes are capable of achieving higher estimation accuracy when more sensor elements are installed. Similarly, the superior performance gain of STARS compared with conventional RIS is also demonstrated. When $\overline{\gamma}=0$dB, the coupled phase-shift model achieves almost the same performance as the independent phase-shift model. When $\overline{\gamma}=15$dB, there is an obvious performance gap between the two phase-shift models, which is gradually reduced as the number of sensor elements increases.

\subsection{Tradeoff Between Passive and Sensor Elements}
In Fig. \ref{fig:element_tradeoff}, to obtain more insights, we plot the root CRB versus the number of passive elements when the total number of the passive and sensor elements is fixed to be $N + N_s = 20$. As can be observed, the more passive elements do not necessarily lead to lower root CRB in this case, which implies that there is a tradeoff between the number of passive elements and the number of sensor elements. Furthermore, for both cases with $\overline{\gamma}=0$ dB and $\overline{\gamma} = 15$ dB, the best sensing performance is achieved when $N=14$ and $N_s = 6$. In other words, installing more passive elements can improve system performance more effectively than installing more sensor elements. This is because more passive elements provide more full-space DoFs for both communication and sensing, while more sensor elements only increase the dimension of the signals used for sensing.

\subsection{Spectrum of MLE}
To verify the effectiveness of optimizing the CRB, the practical MLE of DOAs $\phi_h$ and $\phi_v$ are investigated. Fig. \ref{fig:MLE} demonstrates the normalized spectrum of MLE obtained via the 2D search over a predefined fine angle grid. In the spectrum, the point with the highest value (the brightest point) is corresponding to the estimated DOAs. As can be observed, when $\overline{\gamma} = 0$dB, both independent and coupled phase-shift STARSs can achieve a point-like brightest region around $(120^\circ, 30^\circ)$ in the spectrum, which indicates the high accuracy of estimating DOAs. When $\overline{\gamma}$ increases to $20$dB, the brightest region achieved by STARS is still relatively small. However, when the conventional RIS is exploited, the brightest region becomes very large regardless of the value of $\overline{\gamma}$, which implies a low sensing resolution. Finally, we also demonstrate the spectrum obtained by the STARS with random reflection phase shifts. In this case, we can see that the estimation of DOAs is almost infeasible, which verifies the effectiveness of optimizing the CRB.

\section{Conclusion} \label{sec:conclusion}
A STARS-enabled ISAC system was investigated, where a novel sensing-at-STARS structure and a pair of efficient CRB optimization frameworks for both independent and coupled T\&R phase-shift models were proposed. Numerical results revealed the significant performance gain of CRB and 2D DOAs estimation achieved by the STARS over the conventional RIS, and the slight performance gap between the independent and coupled T\&R phase-shift models in the cases of low communication requirements and sufficient STARS elements. Moreover, compared with employing more sensor elements, installing more passive elements is more efficient in practice to enhance performance and reduce cost. This paper confirmed the effectiveness of employing STARS to simultaneously support high-quality sensing on one side and high-quality communication on the other. Consequently, in practical design, STARS is capable of breaking the physical blockage such as walls and windows, extending the legacy communication or sensing system, and carrying out the other function in a separate space. Finally, the potential of STARS to support dual functions on both sides is also foreseeable, which can be a promising direction for future research.

\section*{Appendix~A\\Maximum Likelihood Estimate of DOAs} \label{appendiex:MLE}
In this appendix, the MLE of DOAs is derived. To facilitate the estimation at the sensors, we assume that the transmit signal $\mathbf{X}$ and the channel $\mathbf{G}$ are known at the sensors, which can be obtained by the wired control link between the BS and the STARS. According to \eqref{eqn:signal_vector}, the vectorized signal over a coherent time block of length $L$ at the sensor can be rewritten as
\begin{equation}
    \mathbf{y}_s = \alpha \boldsymbol{\delta}(\phi_{h}, \phi_{v}) + \mathbf{n}_s,
\end{equation} 
where $\boldsymbol{\delta}(\phi_{h}, \phi_{v}) = \mathrm{vec}(\mathbf{b}(\phi_{h},\phi_{v}) \mathbf{a}^T(\phi_{h},\phi_{v}) \mathbf{\Theta}_r \mathbf{G} \mathbf{X})$.
It can be observed that $\mathbf{y}_s$ is a Gaussian vector with mean $\alpha \boldsymbol{\delta}(\phi_{h}, \phi_{v})$ and variance $\sigma_s \mathbf{I}_{N_\mathrm{S}L}$. Given parameters $\boldsymbol{\xi}$, the likelihood function of $\mathbf{y}_s$ is
\begin{equation}
    f_{\mathbf{y}_s}(\mathbf{y}_s; \boldsymbol{\xi}) = \frac{1}{\sqrt{ (\pi \sigma_s)^{N_{\mathrm{S}} L} }} \mathrm{exp}\left( -\frac{1}{\sigma^2} \| \mathbf{y}_s - \alpha \boldsymbol{\delta}(\phi_{h}, \phi_{v}) \|^2  \right).
\end{equation} 
Thus, the MLE of $\boldsymbol{\xi}$ is given by 
\begin{align} \label{eqn:MLE}
    \hat{\boldsymbol{\xi}}  &= \arg \max_{\boldsymbol{\xi}} f_{\mathbf{y}_s}(\mathbf{y}_s; \boldsymbol{\xi}) = \arg \min_{\boldsymbol{\xi}} \| \mathbf{y}_s - \alpha \boldsymbol{\delta}(\phi_{h}, \phi_{v}) \|^2.
\end{align}
According to \eqref{eqn:MLE}, for any given $\phi_{h}$ and $\phi_{v}$, $\alpha$ can be estimated as 
\begin{equation}
    \hat{\alpha} = \arg \min_{\alpha} \| \mathbf{y}_s - \alpha \boldsymbol{\delta}(\phi_{h}, \phi_{v}) \|^2 = \frac{\boldsymbol{\delta}^H(\phi_{h}, \phi_{v})}{ \| \boldsymbol{\delta}(\phi_{h}, \phi_{v}) \|^2 } \mathbf{y}_s.
\end{equation} 
With $\hat{\alpha}$ at hand, we have 
\begin{align}
    \| \mathbf{y}_s - \hat{\alpha} \boldsymbol{\delta}(\phi_{h}, \phi_{v}) \|^2 = \| \mathbf{y}_s \|^2 - \frac{|\boldsymbol{\delta}^H(\phi_{h}, \phi_{v}) \mathbf{y}_s|^2}{\|\boldsymbol{\delta}(\phi_{h}, \phi_{v})\|^2}.
\end{align}
Thus, the MLE of $\phi_{h}$ and $\phi_{v}$ is given by 
\begin{equation}
    (\hat{\phi}_{h}, \hat{\phi}_{v}) = \arg \max_{\phi_{h}, \phi_{v}} \frac{|\boldsymbol{\delta}^H(\phi_{h}, \phi_{v}) \mathbf{y}_s|^2}{\|\boldsymbol{\delta}(\phi_{h}, \phi_{v})\|^2},
\end{equation}
which can be obtained by exhaustively searching $\phi_{h}$ and $\phi_{v}$ on the fine grids of $[0, \pi]$ and $[-\frac{\pi}{2}, \frac{\pi}{2}]$, respectively.

\section*{Appendix~B\\Derivation of the Fisher Information Matrices} \label{appendiex:FIM}
It can be readily observed that $\mathbf{y}_s$ is a Gaussian observation with the distribution $\mathbf{y}_s \sim \mathcal{CN}(\mathbf{u}, \mathbf{R}_n)$, where $\mathbf{R}_n = \sigma_s^2 \mathbf{I}_{N_sL}$ represents the covariance matrix of $\mathbf{n}_s$. Then, following the same path in \cite[Appendix 3C]{kay1993fundamentals}, the element at the $\ell$-th row and the $p$-th column of $\mathbf{J}_{\boldsymbol{\xi}}$ can be calculated by
\begin{align} \label{eqn:FIM_point}
    [\mathbf{J}_{\boldsymbol{\xi}}]_{\ell,p} = &
    2 \mathrm{Re} \left\{ \frac{\partial \mathbf{u}^H}{\partial \xi_\ell} \mathbf{R}_n^{-1} \frac{\partial \mathbf{u}}{\partial \xi_p} \right\} + \mathrm{tr}\left( \mathbf{R}_n^{-1}  \frac{\partial \mathbf{R}_n}{\partial \xi_\ell} \mathbf{R}_n^{-1} \frac{\partial \mathbf{R}_n}{\partial \xi_p}  \right) \nonumber \\
    = &\frac{2}{\sigma_s^2} \mathrm{Re} \left\{ \frac{\partial \mathbf{u}^H}{\partial \xi_\ell} \frac{\partial \mathbf{u}}{\partial \xi_p} \right\},
\end{align} 
where $\xi_\ell$ denotes the $\ell$-th element of $\boldsymbol{\xi}$. Thus, we first derive the derivative of $\mathbf{u}$ with respect to the unknown parameters. By defining $\mathbf{B} = \mathbf{b}(\phi_{h},\phi_{v}) \mathbf{a}^T(\phi_{h},\phi_{v})$, we have
\begin{align}
    &\frac{\partial \mathbf{u}}{\partial \boldsymbol{\phi}} = [ \alpha \mathrm{vec}(\dot{\mathbf{B}}_{\phi_{h}} \mathbf{\Theta}_r \mathbf{G} \mathbf{X} ), \alpha \mathrm{vec}(\dot{\mathbf{B}}_{\phi_{v}} \mathbf{\Theta}_r \mathbf{G} \mathbf{X} ) ], \\ &\frac{\partial \mathbf{u}}{\partial \tilde{\boldsymbol{\alpha}}} = \mathrm{vec}(\mathbf{B} \mathbf{\Theta}_r \mathbf{G} \mathbf{X}) [1, j],
\end{align}
where
\begin{align}
    \dot{\mathbf{B}}_{\phi_{h}}  = &\frac{\partial \mathbf{B}}{\partial \phi_{h}} = \frac{\partial \mathbf{b}}{{\partial \phi_{h}}} \mathbf{a}^T
    + \mathbf{b} \frac{\partial \mathbf{a}^T}{\partial \phi_{h}} \nonumber \\ = &j \frac{2 \pi}{\lambda_c} \sin \phi_{h} \cos \phi_{v} \left(  \mathrm{diag}(\bar{\mathbf{r}}_{\mathit{X}}) \mathbf{b} \mathbf{a}^T + \mathbf{b} \mathbf{a}^T \mathrm{diag}(\mathbf{r}_{\mathit{X}}) \right). \\   
    \dot{\mathbf{B}}_{\phi_{v}} = &\frac{\partial \mathbf{B}}{\partial \phi_{v}} = \frac{\partial \mathbf{b}}{{\partial \phi_{v}}} \mathbf{a}^T
    + \mathbf{b} \frac{\partial \mathbf{a}^T}{\partial \phi_{v}} \nonumber \\
    = &j \frac{2 \pi}{\lambda_c} \cos \phi_{h} \sin \phi_{v} \left( \mathrm{diag}(\bar{\mathbf{r}}_{\mathit{X}}) \mathbf{b} \mathbf{a}^T + \mathbf{b} \mathbf{a}^T \mathrm{diag}(\mathbf{r}_{\mathit{X}}) \right) \nonumber \\
    &- j \frac{2 \pi}{\lambda_c} \cos \phi_{v} \mathbf{b} \mathbf{a}^T \mathrm{diag}(\mathbf{r}_{\mathit{Z}}),
\end{align}
In the above formulas, we drop $\phi_{h}$ and $\phi_{v}$ in $\mathbf{b}(\phi_{h}, \phi_{v})$ and $\mathbf{a}(\phi_{h}, \phi_{v})$ for notational convenience. Then, the matrix $\mathbf{J}_{\boldsymbol{\phi} \boldsymbol{\phi}}$ can be further partitioned as
\begin{equation}
    \mathbf{J}_{\boldsymbol{\phi} \boldsymbol{\phi}} = \begin{bmatrix}
        J_{\phi_{h} \phi_{h}} &J_{\phi_{h} \phi_{v}} \\
        J_{\phi_{h} \phi_{v}} & J_{\phi_{v} \phi_{v}}
    \end{bmatrix}.
\end{equation}
According to \eqref{eqn:FIM_point}, the entries $J_{\phi_l \phi_p}, \forall l, p \in \{h, v\}, $ of the matrix $\mathbf{J}_{\boldsymbol{\phi} \boldsymbol{\phi}}$
can be calculated as follows:
\begin{align} \label{eqn:FIM_element}
    J_{\phi_l \phi_p} & = \frac{2}{\sigma_s^2} \mathrm{Re} \left\{ \alpha^* \mathrm{vec} ( \dot{\mathbf{B}}_{\phi_l} \mathbf{\Theta}_r \mathbf{G} \mathbf{X})^H \alpha \mathrm{vec} ( \dot{\mathbf{B}}_{\phi_p} \mathbf{\Theta}_r \mathbf{G} \mathbf{X} ) \right\} \nonumber \\
    & =  \frac{2 |\alpha|^2 L }{\sigma_s^2}   \mathrm{Re} \left\{ \mathrm{tr}( \dot{\mathbf{B}}_{\phi_p} \mathbf{\Theta}_r \mathbf{G} \mathbf{R}_x \mathbf{G}^H \mathbf{\Theta}_r^H \dot{\mathbf{B}}_{\phi_l}^H  ) \right\}.
\end{align}
Next, the matrices $\mathbf{J}_{\boldsymbol{\phi} \tilde{\boldsymbol{\alpha}}}$ and $\mathbf{J}_{\tilde{\boldsymbol{\alpha}} \tilde{\boldsymbol{\alpha}}}$
are derived as follows:
\begin{align}
    \mathbf{J}_{\boldsymbol{\phi} \tilde{\boldsymbol{\alpha}}} &= \frac{2}{\sigma_s^2} \mathrm{Re} \left(  \begin{bmatrix}
        \alpha^* \mathrm{vec}(\dot{\mathbf{B}}_{\phi_{h}} \mathbf{\Theta}_r \mathbf{G} \mathbf{X} )^H \\
        \alpha^* \mathrm{vec}(\dot{\mathbf{B}}_{\phi_{v}} \mathbf{\Theta}_r \mathbf{G} \mathbf{X} )^H
    \end{bmatrix} \!\! \mathrm{vec}(\mathbf{B} \mathbf{\Theta}_r \mathbf{G} \mathbf{X}) [1, j] \right)  \nonumber \\
    & = \frac{2 L}{\sigma_s^2} \mathrm{Re} \left( \begin{bmatrix}
        \alpha^* \mathrm{tr}\left( \mathbf{B} \mathbf{\Theta}_r \mathbf{G} \mathbf{R}_x \mathbf{G}^H \mathbf{\Theta}_r^H \dot{\mathbf{B}}_{\phi_{h}}^H \right) \\
        \alpha^* \mathrm{tr}\left( \mathbf{B} \mathbf{\Theta}_r \mathbf{G} \mathbf{R}_x \mathbf{G}^H \mathbf{\Theta}_r^H \dot{\mathbf{B}}_{\phi_{v}}^H \right)
    \end{bmatrix} [1, j] \right), \\
    \mathbf{J}_{\tilde{\boldsymbol{\alpha}} \tilde{\boldsymbol{\alpha}}} &= \frac{2}{\sigma_s^2} \mathrm{Re} \left(  
        ( \mathrm{vec}(\mathbf{B} \mathbf{\Theta}_r \mathbf{G} \mathbf{X}) [1, j] )^H
        \mathrm{vec}(\mathbf{B} \mathbf{\Theta}_r \mathbf{G} \mathbf{X}) [1, j] \right) \nonumber \\ 
    &= \frac{2}{\sigma_s^2} \mathrm{Re} \left( [1, j]^H [1, j] \mathrm{vec}(\mathbf{B} \mathbf{\Theta}_r \mathbf{G} \mathbf{X})^H \mathrm{vec}(\mathbf{B} \mathbf{\Theta}_r \mathbf{G} \mathbf{X}) \right) \nonumber \\
    & = \frac{2 L}{\sigma_s^2} \mathbf{I}_2 \mathrm{tr} \left( \mathbf{B} \mathbf{\Theta}_r \mathbf{G} \mathbf{R}_x \mathbf{G}^H \mathbf{\Theta}_r^H \mathbf{B}^H \right).
\end{align}

\section*{Appendix~C\\Proof of Proposition \ref{proposition:optimal_phase}} \label{appendix:optimal_phase}

For any given $\tilde{\boldsymbol{\beta}}_t$ and $\tilde{\boldsymbol{\beta}}_r$, the optimization problem with respect to $\tilde{\mathbf{q}}_t$ and $\tilde{\mathbf{q}}_r$
is given by 
\begin{subequations}
    \begin{align}
        \min_{\tilde{\mathbf{q}}_t, \tilde{\mathbf{q}}_r} \quad &
        \mathrm{Re}( \tilde{\boldsymbol{\upsilon}}_t^H \tilde{\mathbf{q}}_t ) + \mathrm{Re}( \tilde{\boldsymbol{\upsilon}}_r^H \tilde{\mathbf{q}}_r )\\
        \label{constraint:appdeiex_1}
        \mathrm{s.t.} \quad & [\tilde{\mathbf{q}}_r]_n = j [\tilde{\mathbf{q}}_t]_n \text{ or } [\tilde{\mathbf{q}}_r]_n = -j [\tilde{\mathbf{q}}_t]_n, \forall n, \\
        & |[\tilde{\mathbf{q}}_t]_n| = 1, |[\tilde{\mathbf{q}}_r]_n| = 1, \forall n,
    \end{align}
\end{subequations} 
where constraint \eqref{constraint:appdeiex_1} is transformed from the coupled T\&R phase-shift constraint \eqref{constraint:tilde_theta_2}.
It can be observed this problem is a separable optimization problem. 
In other words, each pair of $(\tilde{q}_{t,n}, \tilde{q}_{r,n})$ can be optimized individually, and the related optimization problem is given by
\begin{subequations}
    \begin{align}
        \min_{\tilde{q}_{t,n}, \tilde{q}_{r,n}} \quad &
        \mathrm{Re}( \tilde{\upsilon}_{t,n}^* \tilde{q}_{t,n} ) + \mathrm{Re}( \tilde{\upsilon}_{r,n}^* \tilde{q}_{r,n} )\\
        \label{constraint:appdeiex_2}
        \mathrm{s.t.} \quad & \tilde{q}_{r,n} = j \tilde{q}_{t,n} \text{ or } \tilde{q}_{r,n} = -j \tilde{q}_{t,n}, \\
        & |\tilde{q}_{t,n}| = 1, |\tilde{q}_{r,n}| = 1.
    \end{align}
\end{subequations}
Substituting the constraint \eqref{constraint:appdeiex_2} into the objective function, the above problem can be further simplified as $\min_{|\tilde{q}_{t,n}| = 1}
\mathrm{Re}\left( (\tilde{\upsilon}_{t,n}^* \pm j \tilde{\upsilon}_{r,n}^*) \tilde{q}_{t,n} \right)$, 
where the factor $(\tilde{\upsilon}_{t,n}^* + j \tilde{\upsilon}_{r,n}^*)$ is for the case $\tilde{q}_{r,n} = j \tilde{q}_{t,n}$ 
and the factor $(\tilde{\upsilon}_{t,n}^* - j \tilde{\upsilon}_{r,n}^*)$ is for the case $\tilde{q}_{r,n} = -j \tilde{q}_{t,n}$.
It is clear that the optimal $\tilde{q}_{t,n}$ is given by
\begin{equation}
    \tilde{q}_{t,n} = e^{j \left(\pi - \angle (\tilde{\upsilon}_{t,n}^* \pm j \tilde{\upsilon}_{r,n}^*) \right)}.
\end{equation}
By matching the optimal $\tilde{q}_{t,n}$ with the cases of $\tilde{q}_{r,n} = j \tilde{q}_{t,n}$ and $\tilde{q}_{r,n} = -j \tilde{q}_{t,n}$, the solutions in \eqref{eqn:optimal_q}
can be obtained, which completes the proof.

\section*{Appendix~D\\Proof of Proposition \ref{proposition:optimal_amplitude}} \label{appendix:optimal_amplitude}
For any given $\tilde{\mathbf{q}}_t$ and $\tilde{\mathbf{q}}_r$, the optimization problem with respect to $\tilde{\boldsymbol{\beta}}_t$ and $\tilde{\boldsymbol{\beta}}_r$ is given by 
\begin{subequations}
    \begin{align}
        \min_{\tilde{\boldsymbol{\beta}}_t, \tilde{\boldsymbol{\beta}}_r} \quad &
        \mathrm{Re}( \breve{\boldsymbol{\upsilon}}_t^H \tilde{\boldsymbol{\beta}}_t  ) + \mathrm{Re}( \breve{\boldsymbol{\upsilon}}_r^H \tilde{\boldsymbol{\beta}}_r  )\\
        \mathrm{s.t.} \quad & \tilde{\beta}_{t,n}^2 + \tilde{\beta}_{r,n}^2 = 1, 0 \le \tilde{\beta}_{t,n}, \tilde{\beta}_{r,n} \le 1, \forall n,
    \end{align}
\end{subequations} 
which is also a separable problem. The separated problem related to $\tilde{\beta}_{t,n}$ and $\tilde{\beta}_{r,n}$ is given by 
\begin{subequations} \label{problem:separate_phase}
    \begin{align}
        \min_{\tilde{\beta}_{t,n}, \tilde{\beta}_{r,n}} \quad &
        g_n = \mathrm{Re}( \breve{\upsilon}_{t,n}^* \tilde{\beta}_{t,n}  ) + \mathrm{Re}( \breve{\upsilon}_{r,n}^* \tilde{\beta}_{r,n}  )\\
        \mathrm{s.t.} \quad & \tilde{\beta}_{t,n}^2 + \tilde{\beta}_{r,n}^2 = 1, 0 \le \tilde{\beta}_{t,n}, \tilde{\beta}_{r,n} \le 1,
    \end{align}
\end{subequations}
where $\tilde{\upsilon}_{i,n}, \forall i \in \{t,r\},$ denotes the $n$-th entry of $\breve{\boldsymbol{\upsilon}}_i$.  
Since $\tilde{\beta}_{i,n}, \forall i \in \{t,r\},$ is real-valued, the objective function can be further simplified as $a_n \tilde{\beta}_{t,n} + b_n \tilde{\beta}_{r,n}$,
where $a_n = |\breve{\upsilon}_{t,n}^*| \cos (\angle \breve{\upsilon}_{t,n}^*)$ and $b_n = |\breve{\upsilon}_{r,n}^*| \cos (\angle \breve{\upsilon}_{r,n}^*)$.
In this case, problem \eqref{problem:separate_phase} is essentially to find the minimum value of the real-valued function $a_n \tilde{\beta}_{t,n} + b_n \tilde{\beta}_{r,n}$ 
on the unit circle $\tilde{\beta}_{t,n}^2  + \tilde{\beta}_{r,n}^2 = 1$ in the first quadrant. To solve it, we transform it into the polar coordinate system 
by setting $\tilde{\beta}_{t,n} = \sin \omega_n$ and $\tilde{\beta}_{r,n} = \cos \omega_n$ with $\omega_n \in [0, \frac{1}{2}\pi]$. Then, the objective function can be rewritten as 
\begin{align}
    g_n &= a_n \sin \omega_n + b_n \cos \omega_n \nonumber \\
    &\overset{(a)}{=} \sqrt{a_n^2 + b_n^2} \left( \cos \psi_n \sin \omega_n + \sin \psi_n \cos \omega_n \right) \nonumber \\
    &= \sqrt{a_n^2 + b_n^2} \sin (\omega_n + \psi_n),
\end{align}
where the equality $(a)$ is achieved by defining $\cos \psi_n = \frac{a_n}{\sqrt{a_n^2 + b_n^2}}$ and $\sin \psi_n = \frac{b_n}{\sqrt{a_n^2 + b_n^2}}$. As a consequence, the problem is to find the minimum value of $\sin (\omega_n + \psi_n)$ with respect to $\omega_n$ in the interval $[0, \frac{1}{2}\pi]$. Thus, the optimal $\omega_n$ in \eqref{eqn:optimal_omega} can be readily obtained. Based on this, the optimal $\tilde{\beta}_{t,n}$ and $\tilde{\beta}_{r,n}$ can also be obtained, which completes the proof.

\bibliographystyle{IEEEtran}
\begin{spacing}{1.0}
\bibliography{reference/mybib}
\end{spacing}

\end{document}